\documentclass[11pt]{article}
\pdfoutput=1

\usepackage{amsfonts}
\usepackage{amsmath}
\usepackage{amssymb}
\usepackage{mathrsfs}
\usepackage{amsthm}
\usepackage[english]{babel}
\usepackage[latin1]{inputenc}
\usepackage[T1]{fontenc}
\usepackage{epsfig}
\usepackage{textcomp}
\usepackage{graphicx}
\usepackage{epstopdf}
\usepackage{pgfplots}
\pgfplotsset{compat=newest}
\usepackage{booktabs}

\numberwithin{equation}{section}
\newcommand{\clearemptydoublepage}{\newpage{\pagestyle{empty}\cleardoublepage}}
\newtheorem{lemma}{Lemma}
\newtheorem{proposition}{Proposition}
\newtheorem{corollary}{Corollary}

\def\a{\alpha}
\def\b{\beta}

\def\d{\delta}

\def\ep{\epsilon}
\def\vep{\varepsilon}
\def\z{\zeta}
\def\e{\eta}

\def\vth{\vartheta}

\def\la{\lambda}
\def\La{\Lambda}
\def\m{\mu}
\def\n{\nu}
\def\x{\xi}
\def\vp{\varpi}
\def\r{\rho}

\def\s{\sigma}

\def\t{\tau}

\def\vf{\varphi}
\def\ch{\chi}
\def\ps{\psi}
\def\o{\omega}

\def\ta{\tilde{\alpha}}

\def\tt{\tilde{\tau}}
\def\tp{\tilde{\pi}}
\def\tvp{\tilde{\varpi}}

\newcommand{\ti}[1]{\tilde{#1}}

\newcommand{\Xd}{X^{\cdot}}
\newcommand{\xid}{\xi^{\cdot}}
\newcommand{\chd}{\chi^{\cdot}}

\newcommand{\xd}{x^{\cdot}}
\newcommand{\vfd}{\varphi^{\cdot}}

\newcommand{\bn}{\bar{n}}

\newcommand{\bvf}{\bar{\varphi}}
\newcommand{\bvfd}{\bar{\varphi}^{\cdot}}

\newcommand{\hxi}{\hat{\xi}}
\newcommand{\hxid}{\hat{\xi}^{\cdot}}

\newcommand{\h}[1]{\hat{#1}}
\newcommand{\hp}{\hat{\pi}}

\providecommand{\abs}[1]{\lvert#1\rvert}

\newcommand{\de}{\partial}
\newcommand{\half}{\frac{1}{2}}

\newcommand{\p}{\prime}

\newcommand{\beq}{\begin{equation}}
\newcommand{\eeq}{\end{equation}}
\newcommand{\beqnn}{\begin{equation*}}
\newcommand{\eeqnn}{\end{equation*}}
\newcommand{\bea}{\begin{eqnarray}}
\newcommand{\eea}{\end{eqnarray}}

\newcommand{\noi}{\noindent}

\newcommand{\xim}{\xi^{\mu}}
\newcommand{\chm}{\chi^{\mu}}

\newcommand{\xm}{x^{\mu}}

\newcommand{\mbf}[1]{\mathbf{#1}}
\newcommand{\mbfG}{\mathbf{G}}

\newcommand{\tmbfg}{\tilde{\mathbf{g}}}
\newcommand{\tmbfK}{\tilde{\mathbf{K}}}
\newcommand{\tmbfG}{\tilde{\mathbf{G}}}
\newcommand{\tmbfT}{\tilde{\mathbf{T}}}
\newcommand{\mcal}[1]{\mathcal{#1}}
\newcommand{\mcalT}{\mathcal{T}}
\newcommand{\mcalC}{\mathcal{C}}
\newcommand{\mcalI}{\mathcal{I}}

\newcommand{\mcalA}{\mathcal{A}}
\newcommand{\mcalB}{\mathcal{B}}
\newcommand{\mcalF}{\mathcal{F}}

\newcommand{\mscr}[1]{\mathscr{#1}}
\newcommand{\mscrT}{\mathscr{T}}

\newcommand{\dez}{\de_{z}}
\newcommand{\dey}{\de_{y}}
\newcommand{\dezq}{\de_{z}^{2}}

\newcommand{\ded}[1]{\de_{#1}}

\newcommand{\dem}{\de_{\m}}

\newcommand{\boxf}{\Box_{4}}

\newcommand{\boxfi}{\Box_{5}}
\newcommand{\boxs}{\Box_{6}}

\newcommand{\dexi}{\de_{\xi}}
\newcommand{\dehxi}{\de_{\hat{\xi}}}

\newcommand{\dexim}{\de_{\xi^{\m}}}

\newcommand{\deximxin}{\de_{\xi^{\m}} \de_{\xi^{\n}}}

\newcommand{\dechm}{\de_{\chi^{\m}}}

\newcommand{\dechmchn}{\de_{\chi^{\m}} \de_{\chi^{\n}}}

\newcommand{\evb}{\Big\rvert_{X^{\cdot} = \bar{\varphi}^{\cdot}(\xi^{\cdot})}}

\newcommand{\evbhsi}{\Big\rvert_{\bar{\varphi}^{\cdot}(\hat{\xi}^{\cdot})}}
\newcommand{\evbhsin}{\Big\rvert_{\bar{\varphi}_{_{[n]}}^{\cdot}(\hat{\xi}^{\cdot})}}
\newcommand{\evbhsinf}{\Big\rvert_{\bar{\varphi}_{_{\infty}}^{\cdot}(\hat{\xi}^{\cdot})}}

\newcommand{\evbbhsi}{\bigg\rvert_{\bar{\varphi}^{\cdot}(\hat{\xi}^{\cdot})}}
\newcommand{\evbbhsinf}{\bigg\rvert_{\bar{\varphi}_{_{\infty}}^{\cdot}(\hat{\xi}^{\cdot})}}

\newcommand{\Msf}{M_{6}^{4}}
\newcommand{\Mft}{M_{5}^{3}}
\newcommand{\Mfs}{M_{4}^{2}}

\newcommand{\Zp}{Z^{\prime}}

\newcommand{\Yp}{Y^{\prime}}

\newcommand{\Sp}{S^{\prime}}
\newcommand{\Spq}{{S^{\prime}}^2}
\newcommand{\Spp}{S^{\prime \prime}}

\newcommand{\mZ}{\mathscr{Z}}
\newcommand{\mY}{\mathscr{Y}}

\newcommand{\bla}{\bar{\la}}
\newcommand{\blac}{\bar{\la}_{c}}
\newcommand{\dla}{\d \! \la}

\newcommand{\ltp}{l_2^{\shortparallel}}
\newcommand{\ltn}{l_2^{\perp}}
\newcommand{\lo}{l_1}

\newcommand{\dvf}{\delta \! \varphi}
\newcommand{\hdvf}{\delta \! \hat{\varphi}}

\newcommand{\dvfn}{\delta \! \varphi_{\!_{\perp}}}
\newcommand{\hdvfn}{\delta \! \hat{\varphi}_{\!_{\perp}}}
\newcommand{\hdvfnin}{\delta \! \hat{\varphi}_{\!_{\perp}}^{\infty}}
\newcommand{\hdvfninp}{\delta \! \hat{\varphi}_{\!_{\perp}}^{\infty \, \prime}}
\newcommand{\dvfp}{\delta \! \varphi_{\shortparallel}}
\newcommand{\hdvfp}{\delta \! \hat{\varphi}_{\shortparallel}}

\newcommand{\dvfzgi}{\delta \! \varphi^{z}_{gi}}
\newcommand{\dvfygi}{\delta \! \varphi^{y}_{gi}}

\newcommand{\bv}{\bar{v}}
\newcommand{\dv}{\delta v}

\newcommand{\dvn}{\delta v_{\!_{\perp}}}
\newcommand{\dhvn}{\delta \hat{v}_{\!_{\perp}}}

\newcommand{\dhvsp}{\delta \hat{v}_{\shortparallel}}

\begin{document}

\date{\mbox{}}

\title{
\vspace{-2.0cm}
\vspace{2.0cm}
{\bf \huge Perturbations of Nested Branes With Induced Gravity}
 \\[8mm]
}

\author{
Fulvio Sbis\`a$^{1,2}$\thanks{fulvio.sbisa@port.ac.uk} \phantom{i}and Kazuya Koyama$^{1}$\thanks{kazuya.koyama@port.ac.uk}
\\[8mm]
\normalsize\it
$^1$ Institute of Cosmology \& Gravitation, University of Portsmouth,\\
\normalsize\it Dennis Sciama Building, Portsmouth, PO1 3FX, United Kingdom\vspace{.5cm} \\
\normalsize\it
$^2\,$ Dipartimento di Fisica dell'Universit\`a di Milano\\
       {\it Via Celoria 16, I-20133 Milano} \\
       and \\
\normalsize\it
       INFN, Sezione di Milano, \\
       {\it Via Celoria 16, I-20133 Milano}\vspace{.5cm} \\[.3em]
}

\maketitle

\setcounter{page}{1}
\thispagestyle{empty}

\begin{abstract}

\noindent We study the behaviour of weak gravitational fields in models where a 4D brane is embedded inside a 5D brane equipped with induced gravity, which in turn is embedded in a 6D spacetime. We consider a specific regularization of the branes internal structures where the 5D brane can be considered thin with respect to the 4D one. We find exact solutions corresponding to pure tension source configurations on the thick 4D brane, and study perturbations at first order around these background solutions. To perform the perturbative analysis, we adopt a bulk-based approach and we express the equations in terms of gauge invariant and master variables using a 4D scalar-vector-tensor decomposition. We then propose an ansatz on the behaviour of the perturbation fields when the thickness of the 4D brane goes to zero, which corresponds to configurations where gravity remains finite everywhere in the thin limit of the 4D brane. We study the equations of motion using this ansatz, and show that they give rise to a consistent set of differential equations in the thin limit, from which the details of the internal structure of the 4D brane disappear. We conclude that the thin limit of the ``ribbon'' 4D brane inside the (already thin) 5D brane is well defined (at least when considering first order perturbations around pure tension configurations), and that the gravitational field on the 4D brane remains finite in the thin limit. We comment on the crucial role of the induced gravity term on the 5D brane.
 
\end{abstract}

\smallskip

\section{Introduction}

The Cosmological Constant problem (CC) and the late time acceleration of the universe (LTA) are two of the most compelling problems in contemporary theoretical physics. Despite it has been noted long ago \cite{WeinbergCC} that there is a difference of 60-120 orders of magnitude (depending on the cut-off of the quantum theory) between the theoretical estimates for the vacuum energy and its observed value, a fully convincing explanation of this mismatch is still missing. The recent discovery that the cosmological observations are best fitted by a model where the cosmological constant $\La$ is non-zero \cite{Riess98, Perlmutter98} renders the problem even more puzzling. In fact, along with explaining why $\La$ is small compared to the theoretical predictions (\emph{old} cosmological constant problem), it is now necessary to explain also why it is non-zero and yet extremely fine-tuned (\emph{new} cosmological constant problem). It is fair to say that it is in principle not obvious that these two problems have a common origin, since the reason why the universe is accelerating may be independent from the reason why the semiclassical effect of the vacuum energy is so different from what we expect. 

A very promising framework to address these problems is offered by the braneworld paradigm (see \cite{Akama82,RubakovShaposhnikov83} for early proposals), which is appealing from the point of view of high energy physics since the existence of branes and of extra dimensions is an essential ingredient in string theory (see for example \cite{PolchinskiBook}). Concerning the cosmological constant problem, it has been noted that braneworld models with infinite volume extra dimensions can bypass Weinberg's no-go theorem \cite{ArkaniHamed:2002fu,Dvali:2002pe}, and more specifically that (in the case of codimension higher than one) they somehow can act as a high-pass filter on the wavelength of gravitational sources, effectively ``degravitating'' sources which are nearly constant with respect to a characteristic length of the model \cite{Dvali:2007kt}. In particular, branes of codimension two have the well known property that pure vacuum energy does not produce curvature on the brane itself, but merely curves the extra dimensions: this opens up the possibility of having models where, in presence of matter, the effective vacuum energy on the brane dynamically relaxes to a very small value after a phase transition happens (\emph{self-tuning}) (see \cite{RubakovShaposhnikov83bis} and \cite{Koyama:2007rx,Nilles:2003km} for a review). Furthermore, if we believe that the CC problem and the LTA problem are independent, braneworld models with induced gravity offer the possibility of explaining the late time acceleration in a geometrical way, since they generically admit cosmological solutions which exhibit the phenomenon of self-acceleration.

Overall we can say that codimension-2 branes with infinite volume extra dimensions and induced gravity are a very interesting framework to address the CC problem and the LTA problem. On the other hand, they suffer from the notorious shortcomings that the thin limit of a brane is not well-defined if the codimension is higher than one \cite{GerochTraschen}, and that the brane-to-brane propagator of the gravitational field diverges when we send to zero the thickness of the brane \cite{Cline:2003ak, Vinet:2004bk} (unless we allow for Gauss-Bonnet terms in the bulk action \cite{Bostock:2003cv}). More worryingly, braneworld models with induced gravity are often plagued by the presence of ghosts. While in the (codimension-1) DGP model \cite{DGP00} ghost modes propagate only around self-accelerated solutions \cite{Koyama:2005tx, Gorbunov:2005zk, LutyPorratiRattazzi, NicolisRattazzi, DGPspectereoscopy, Koyama:2007za}, higher codimension generalizations of the same model seem to suffer from the presence of ghost even around solutions where the bulk is flat and the brane is straight \cite{Dubovsky:2002jm} (however see \cite{Berkhahn:2012wg}). This seems anyway to be a regularization dependent property, since specific regularization procedures seem to render the model ghost-free, at least when flat solutions are concerned \cite{Kolanovic:2003am}. From another point of view, it has been shown that the self-accelerating solutions of the (codimension-1) DGP model fit the cosmological data significantly worse than the $\La$CDM model \cite{MaartensMajerotto06, Rydbeck:2007gy}, and in fact even if they were ghost-free they would be observationally ruled out \cite{Fang:2008kc}.

This state of affairs prompts us to look for higher codimension extensions of the DGP model which preserve its good aspects while being free from its shortcomings. In fact, it can be expected that increasing the codimension would improve the agreement between the theoretical predictions of self-accelerating cosmological solutions and the observational data \cite{Agarwal:2009gy}. An interesting direction to explore is to consider elaborate constructions with more than one brane (as for example intersecting brane scenarios \cite{Corradini:2007cz, Corradini:2008tu}), hoping that the interplay between the branes may provide a mechanism to get rid of the ghosts. A very interesting proposal in this sense appeared few years ago, in which a $D$-dimensional bulk ($D \geq 6$) contains a sequence of branes of increasing dimensionality recursively embedded one into the other, and every brane is equipped with an induced gravity term (the Cascading DGP model \cite{deRham:2007xp}). It has been shown \cite{deRham:2007xp,deRham:2010rw} in fact that, in the 6D realization of this set-up, there exists a critical value $\blac$ for the tension of the codimension-2 brane such that first order perturbations around pure tension backgrounds contain a ghost mode or are ghost-free depending on the fact that the background tension $\bla$ is smaller or bigger than the critical tension. Furthermore, it has been claimed that the induced gravity term on the codimension-1 brane renders finite the (codimension-2) brane-to-brane propagator \cite{deRham:2007xp,deRham:2007rw}, thereby regularizing gravity. Both these properties (gravity regularization and existence of a critical tension) are very interesting and unexpected.

The purpose of this paper is to study the behaviour of weak gravitational fields in set-ups where matter is confined on a codimension-2 brane which is embedded inside a codimension-1 brane with induced gravity, focusing on the thin limit of the codimension-2 brane. More precisely, we consider (background) configurations where the source on the codimension-2 brane has the form of pure tension, and we study the behaviour of the gravitational field at first order in perturbations around these solutions. The main aim of our analysis is to understand geometrically the mechanism of gravity regularization and the role of the induced gravity term. Moreover, we want to develop a formalism to study perturbations in the nested branes with induced gravity set-ups which, in a future perspective, can be applied to other background configurations and to the Cascading DGP model (by adding an induced gravity term on the codimension 2 brane). The paper is therefore structured as follows: in section \ref{Nested branes with induced gravity} we introduce the specific regularization choice we use to study the nested brane system, and clarify the properties of the set-up. In section \ref{Perturbations around pure tension solutions} we derive the pure tension solutions, and consider first order perturbations using a bulk-based approach and the 4D scalar-vector-tensor decomposition. In section \ref{The equations of motion for the perturbations} we derive the equations of motion for the gauge invariant variables of the system. We introduce the master variables and study the case of a pure tension perturbation. In section \ref{ThinLimitCod2Brane} we propose an ansatz on the behaviour of the perturbation fields in the thin limit of the cod-2 brane. We derive the thin limit equations of motion and discuss their consistency. Finally, we present our conclusions in section \ref{Conclusions}.

\textbf{Conventions}: For metric signature, connection, covariant derivative, curvature tensors and Lie derivative we follow the conventions of Misner, Thorne and Wheeler \cite{MisnerThorneWheeler}. The metric signature is the ``mostly plus'' one, and we def\mbox{}ine symmetrization and antisymmetrization \emph{without} normalization. 6D indices are denoted by capital letters, so run from 0 to 5; 5D indices are denoted by latin letters, and run from 0 to 4, while 4D indices are denoted by greek letters and run from 0 to 3. The only exception is that the letters $i$, $j$ and $k$ indicate 2D indices which run on the extra dimensions $z$ and $y$. In general, quantities pertaining to the cod-1 brane are denoted by a tilde $\tilde{\phantom{a}}$, while quantities pertaining to the cod-2 brane are denoted by a superscript $\phantom{a}^{_{(4)}}$. Abstract tensors are indicated with bold-face letters, while quantities which have more than one component but are not tensors (such as coordinates $n$-tuples for example) are expressed in an abstract way replacing every index with a dot. 
When studying perturbations, the symbol $\simeq$ indicates usually that an equality holds at linear order. We use throughout the text the (Einstein) convention of implicit summation on repeated indices, and we use unit of measure where the speed of light has unitary value $c=1$.

\section{Nested branes with induced gravity}
\label{Nested branes with induced gravity}

Configurations where a brane is embedded inside another brane of higher dimensionality (see \cite{Morris:1997hj,Edelstein:1997ej} for a field theory realization), have been already studied for example in \cite{Gregory:2001xu,Gregory:2001dn} (without induced gravity terms) and \cite{Dvali:2006if} (with induced gravity terms) in the context of 5D braneworld models, where extended sources inside the 4D brane were used to investigate the non-perturbative properties of these theories. More recently they have been considered in the context of the Cascading DGP model \cite{deRham:2007xp}, where the existence of a critical tension and the regularization of gravity have been uncovered. In this paper, we consider systems where a (codimension 2) 4D brane is embedded inside a (codimension 1) 5D brane which is in turn embedded in a 6D bulk, and only the cod-1 brane is equipped with an induced gravity term. Therefore, we consider systems which are schematically described by the action
\beq
\label{6D General action}
S = \Msf \int_{\mcal{B}} \!\! d^6 X \, \sqrt{-g} \, R + \Mft \int_{\mcal{C}_1} \!\! d^5 \xi \, \sqrt{-\ti{g}} \, \ti{R} + \int_{\mcal{C}_2} \!\! d^4 \ch \, \sqrt{-g^{(4)}} \, \mscr{L}_{M}
\eeq
where the bulk $\mcal{B}$ is parametrized by the coordinates $\Xd = (z,y,x^{\m})$, the cod-1 brane $\mcal{C}_1$ is parametrized by the coordinates $\xid = (\xi, \xi^{\m})$ and the cod-2 brane $\mcal{C}_2$ is parametrized by the coordinates $\chd$. Here $\ti{\mathbf{g}}$ indicates the metric induced on the cod-1 brane (and $\ti{R}$ is the Ricci scalar built with it), while $\mathbf{g}^{_{(4)}}$ indicates the metric induced on the cod-2 brane and the Lagrangian $\mscr{L}_{M}$ describes the matter localized on the cod-2 brane.

\subsection{Regularization choice and thickness hierarchy}
\label{Regularization choice and thickness hierarchy}

To understand if gravity on the cod-2 brane is regularized by the (cod-1) induced gravity term, we should study how the gravitational field behaves when the cod-2 brane becomes thin. However, such an analysis is meaningful only if the concept of ``thin cod-2 brane'' is meaningful (i.e.~if the thin limit is well-defined). If not, the fact that gravity is regularized may be true only when the limit is performed in some specific ways (if any); this would physically mean that gravity is regularized only if the internal structures of the branes have some specific properties. To this effect, it has been proved that the thin limit of a (isolated) brane is not well-defined if its codimension is higher than one \cite{GerochTraschen}. It is reasonable to expect that this fact does not change if we embed a cod-2 brane inside a cod-1 brane, since, beside the freedom to choose the cod-2 internal structure, we now have the additional freedom to choose how the internal structures of the two (cod-1 and cod-2) branes are related one to the other (see \cite{deRham:2007rw} for a related discussion in the Cascading DGP model). Therefore, in absence of a rigorous proof (on the lines of \cite{GerochTraschen}) of the well-definiteness of the thin limit of the nested cod-1 and cod-2 branes, to perform a clean analysis we should consider configurations where both branes are thick.
%

Working with configurations where both branes are thick is however extremely complicated, and probably not doable in practice. To facilitate the analysis, we could consider particular cases in which there is a hierarchy of scales between the two branes, with the hope that this permits to describe the system with a good approximation by considering one of the branes thin (relatively to the other). 
\begin{figure}[htb!]
\centering
\begin{tikzpicture}
\begin{scope}[scale=.9,rotate=30,>=stealth]
\fill[color=green!20!white] (-6,-1.5) rectangle (6,1.5);
\filldraw[fill=red!20!white,draw=black] (0,0) ellipse (2cm and 1cm);
\draw[very thick] (-6,-1.5) -- (6,-1.5);
\draw[very thick] (-6,1.5) -- (6,1.5);
\draw[gray,thin,<->] (0,0) -- node[color=black,anchor=west]{$l_2^{\perp}$} (0,1);
\draw[gray,thin,<->] (0,0) -- node[color=black,anchor=north]{$l_2^{\shortparallel}$} (2,0);
\draw[gray,thin,<->] (5,-1.5) -- node[color=black,anchor=east]{$l_1$} (5,1.5);
\end{scope}
\end{tikzpicture}
\caption[Characteristic scales for the cod-1 and cod-2 branes]{Characteristic scales for the cod-1 brane (green) and the cod-2 brane (ellipse, violet)}
\label{characteristic nested scales}
\end{figure}
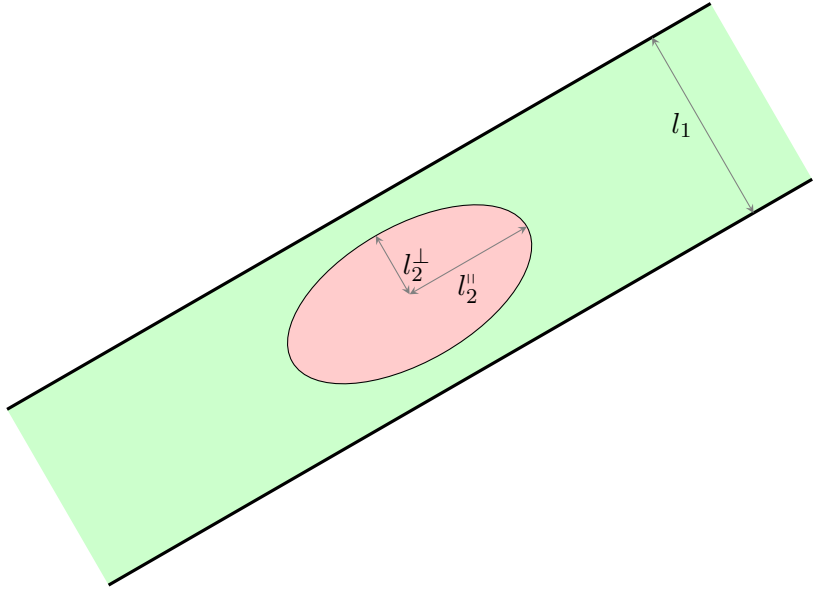
To clarify this point, let's consider in fact the simple schematic description of figure \ref{characteristic nested scales} where the 2D sections of the branes in the extra dimensions are plotted (each point in the figure represents a 4D spacetime): we indicate with $l_1$ the thickness of the cod-1 brane, with $l_2^{\shortparallel}$ the thickness of the cod-2 brane in the parallel directions (from the point of view of the cod-1 brane) and with $l_2^{\perp}$ the thickness of the cod-2 brane in the normal direction. Among the infinite possible choices for the three representative thicknesses, we could consider the class of configurations where the thicknesses satisfy the following hierarchy $\ltp \gg \ltn$, $\ltn \sim \lo$; in this case, we may consider the cod-1 brane to be thin with respect to the cod-2 brane, and we could describe this situation by considering a perfectly thin cod-1 brane, and ask that the matter on the cod-1 brane is distributed only inside a ribbon of width $\sim \ltp$. This situation is definitely appealing since we know that in this case the internal structure of the cod-1 brane does not play a role, and we can study it using the formalism of thin cod-1 branes. Furthermore, in this case there is a clear connection with the (cod-1) DGP model, which is well understood and can serve as a guidance. 

Therefore, in the following we will consider only configurations which satisfy the hierarchy outlined above, and we will describe them assuming that energy and momentum are distributed inside a ``ribbon'' cod-2 brane which lies inside a thin cod-1 brane. Henceforth, we refer to this class of configurations as \emph{nested branes with induced gravity} set-ups. Note that a priori we don't know if the thin limit of a ribbon cod-2 brane inside an already thin cod-1 brane is well defined or not. This is in fact a very important point to establish, and we will address it in our analysis. If this (second) thin limit is well defined, it is possible to work with a thin cod-2 brane and forget the internal structure of the cod-2 brane as well, thereby simplifying further the analysis.

\subsection{The set-up}

Before turning to the analysis of the equations of motion, it is useful to spell out clearly the properties of the nested branes with induced gravity set-up; this also gives us the chance to clarify the notation.

\subsubsection{The geometric set-up}
\label{The geometric set-up}

We assume that the cod-1 brane $\mcal{C}_1$ divides the 6D spacetime in two disconnected pieces which are (globally) diffeomorphic, and we assume that a $\mathbb{Z}_{2}$ symmetry holds across the cod-1 brane. We furthermore assume that there is a 4D submanifold $\mcal{C}_2$ embedded inside the cod-1 brane, which likewise divides the cod-1 brane in two patches whose common boundary is $\mcal{C}_2$. We assume that matter and tension are confined inside the cod-1 brane $\mcal{C}_1$ and localized around the 4D submanifold $\mcal{C}_2$, and we assume that a $\mathbb{Z}_{2}$ symmetry across $\mcal{C}_2$ holds inside the cod-1 brane. More specifically, we distinguish between a physical (thick) cod-2 brane, inside which matter and tension are confined (the ``ribbon'' cod-2 brane), and a mathematical (thin) cod-2 brane ($\mcal{C}_2$), with respect to which the $\mathbb{Z}_{2}$ symmetry is imposed. When the thin limit of the cod-2 brane is performed, the physical brane coincides with the mathematical one. Note that, differently from the original formulation of the Cascading DGP model \cite{deRham:2007xp}, we do \emph{not} impose a $\mathbb{Z}_{2} \times \mathbb{Z}_{2}$ symmetry to hold in each of the two mirror copies which constitute the bulk. In fact, the presence of a $\mathbb{Z}_{2}$ symmetry inside the cod-1 brane does not imply that a double $\mathbb{Z}_{2}$ symmetry holds in each of the two 6D mirror copies.

Since $\mcal{C}_1$ and $\mcal{C}_2$ are submanifolds of the 6D ambient space, they may be considered as separate manifolds, each one equipped with its atlas of reference systems plus an embedding function which describes how they are embedded in the ambient space. The position of the cod-1 brane in the bulk is described by the embedding function $\vfd$ whose component expression is $\varphi^{A}(\xi^{a})$, while the position of the (mathematical) cod-2 brane inside the cod-1 brane is described by the embedding function $\ti{\a}^{\cdot}$ whose expression in coordinates is $\ti{\a}^{a}(\chi^{\m})\,$. The bulk metric $\mathbf{g}$ induces on the codimension-1 brane the metric $\ti{\mathbf{g}} \equiv \varphi_{\star} \big( \mathbf{g} \big)$, where $\varphi_{\star}$ indicates the pullback with respect to the embedding function $\vf^{\cdot}$, and in turn the metric $\ti{\mathbf{g}}$ induces on the codimension-2 brane a metric $\mathbf{g}^{_{(4)}} \equiv \ti{\a}_{\star} \big( \mathbf{\ti{g}} \big)$, where $\ti{\a}_{\star}$ indicates the pullback with respect to the embedding function $\ti{\a}^{\cdot}$.

In the following it will be useful to consider reference systems on the codimension-1 brane which are Gaussian Normal with respect to the (matematical) cod-2 brane: henceforth, we refer to this class of reference systems as codimension-1 Gaussian Normal reference systems. We indicate quantities evaluated in this coordinate systems with an overhat $\h{\phantom{a}}$, and we synthetically indicate the cod-1 GN coordinates as $\hxid \equiv ( \hxi, \chd )$. By construction, we have that \cite{CarrollBook}
\begin{align}
\label{c1GNCgeneral}
\h{g}_{\xi\xi}(\hxi,\chd) &= 1 & \h{g}_{\xi\m}(\hxi,\chd) &= 0
\end{align}
and moreover we have that, choosing a fixed $\hxi$, the 4D tensor $\h{g}_{\m\n}(\hxi,\chd)$ (seen as a function of $\chd$) is the induced metric on the 4D slice characterized by that particular $\hxi$.

\subsubsection{The source set-up}

As we mentioned above, we assume that the energy-momentum tensor present on the cod-1 brane is localized inside the physical (ribbon) cod-2 brane. By ``localized'' we mean that, first of all, there exists a (finite) localization length $l_2$ such that the energy-momentum tensor $\h{T}_{ab}(\hxi, \chd)$ (in cod-1 GNC) vanishes when it is evaluated at a distance $\hxi$ from the cod-2 brane which is bigger than $l_2$ (the length $l_2$ corresponds, in the language of section \ref{Regularization choice and thickness hierarchy}, to the ``parallel'' thickness $\ltp$). Secondly, we ask that the pillbox integration across the cod-2 brane of the normal and mixed components of $\h{T}_{ab}$ vanishes
\beq
\label{angelica}
\int_{-l_2}^{+l_2} d \hxi \,\, \hat{T}_{\xi\xi}(\hxi, \chd) = \int_{-l_2}^{+l_2} d \hxi \,\, \hat{T}_{\xi\m}(\hxi, \chd) = 0
\eeq
which formalizes the idea that momentum does not flow out of the brane.

We define the cod-2 energy-momentum tensor as the 4D tensor $T^{_{(4)}}_{\m\n}(\chd)$ obtained by the pillbox integration of the 4D components of $\h{T}_{ab}$ across the cod-2 brane, so that we have
\beq
\label{cristina}
\int_{-l_2}^{+l_2} d \hxi \,\, \hat{T}_{ab}(\hxi, \chd) = \d_{a}^{\, \, \m} \,\, \d_{b}^{\, \, \n} \,\, T^{(4)}_{\m\n}(\chd)
\eeq
The latter tensor can be considered as the ``would be'' thin limit source configuration if the thin limit description were well-defined. From this point of view, we can consider different configurations $\ti{T}_{ab}$ which correspond to the same $T^{_{(4)}}_{\m\n}$ as different regularizations of the perfectly localized source $T^{_{(4)}}_{\m\n}$. In the following, it will be useful to perform the pillbox integration of the junction conditions across the cod-2 brane. We therefore introduce the notation
\beq
\label{strange equation1}
\int_{-}^{+} \! d\hxi \equiv \int_{-l_2}^{+l_2} \! d\hxi
\eeq
and also, given a quantity $\mscr{Q}(\xi)$ defined on the cod-1 brane, we indicate
\begin{align}
\mscr{Q}\Big\rvert_{+} &\equiv \mscr{Q}\Big\rvert_{\xi = +l_2} & \mscr{Q}\Big\rvert_{-} &\equiv \mscr{Q}\Big\rvert_{\xi = -l_2}
\end{align}
and finally
\beq
\label{strange equation2}
\Big[ \mscr{Q} \Big]_{\pm} \equiv \mscr{Q}\Big\rvert_{\xi = +l_2} - \mscr{Q}\Big\rvert_{\xi = -l_2}
\eeq

\subsubsection{The equations of motion}
\label{The equations of motion}

Since the cod-1 brane is thin, the physical configurations of the theory can be found by solving the Einstein equations in the bulk and by imposing the Israel junction conditions \cite{Israeljc} at the cod-1 brane. As usual for cod-1 branes, to single out a unique solution we have to add boundary conditions at spatial infinity (i.e.~in the extra dimensions); we implicitly assume this in the following. Moreover, because of the $\mathbb{Z}_2$ symmetry which holds across the cod-1 brane, it is enough to choose one of the two mirror copies which constitute the bulk, and to solve the Einstein equations only there (henceforth, with a slight abuse of language we refer to the chosen copy as the ``bulk'' itself).

Since we allow the embedding of the cod-1 brane to be non-trivial, it is useful to describe the geometrical properties of the cod-1 brane using objects which are intrinsic to the brane. In general, the extrinsic geometry of the cod-1 brane is described by the second fundamental form
\beq
\mbf{K} \equiv - \half \, \mcal{L}_{\mbf{n}} \, \mbf{P}
\eeq
where $\mbf{n}$ is the vector normal to the cod-1 brane, $\mbf{P} \equiv \mbf{g} - \mbf{g}(\mathbf{n},\_) \otimes \mbf{g}(\mathbf{n},\_)$ is the first fundamental form 
and $\mcal{L}$ indicates the Lie derivative. We fix the arbitrariness related to the choice of orientation of the brane by choosing the normal vector which points inward the bulk. To obtain an intrinsic object which describes the extrinsic geometry, we can pull-back $\mbf{K}$ to the brane obtaining the extrinsic curvature\footnote{Note that in some references $\mbf{K}$ is called ``extrinsic curvature'' instead of $\tilde{\mathbf{K}}$.} $\tilde{\mathbf{K}}(\xi^{\cdot})$
\beq
\tilde{\mathbf{K}} \equiv \vf_{\star} \big( \mbf{K} \big)
\eeq
which explicitly contains the embedding function. Using the explicit expression for the Lie derivative, we can decompose the extrinsic curvature as follows
\beq
\label{ExtrCurvDecomposition}
\tmbfK(\xi^{\cdot}) = \tmbfK^{[og]}(\xi^{\cdot}) + \tmbfK^{[pg]}(\xi^{\cdot}) + \tmbfK^{[b]}(\xi^{\cdot})
\eeq
where the three contributions read in coordinates
\begin{align}
\label{ExtrCurvOG}
\tilde{K}^{[og]}_{ab}(\xi^{\cdot}) &\equiv - \half \, \frac{\de \varphi^{A}(\xid)}{\de \xi^{a}}
\frac{\de \varphi^{B}(\xid)}{\de \xi^{b}} \, n^L(\xid) \frac{\de \, g_{AB}}{\de X^L}\Big\rvert_{X^{\cdot} = \vf^{\cdot}(\xid)} \\[2mm]
\label{ExtrCurvPG}
\tilde{K}^{[pg]}_{ab}(\xi^{\cdot}) &\equiv \half \, n^A(\xid) \, \frac{\de \varphi^{B}(\xid)}{\de \xi^{(a}} \, \frac{\de \varphi^{L}(\xid)}{\de \xi^{b)}} \frac{\de \, g_{AB}}{\de X^L}\Big\rvert_{X^{\cdot} = \vf^{\cdot}(\xid)}\\[2mm]
\label{ExtrCurvB}
\tilde{K}^{[b]}_{ab}(\xi^{\cdot}) &\equiv n_{L}(\xid) \, \frac{\de^2 \varphi^{L}(\xid)}{\de \xi^{a} \de \xi^{b}}
\end{align}
and we defined $n_M(\xid) = g_{LM}(\vfd(\xid)) \, n^M(\xid)$. Note that the vectors $\mathbf{v}_{(a)}$ defined in coordinates by
\beq
\label{parallelvectorsgeneral}
v_{(a)}^A(\xid) \equiv \bigg\{ \frac{\de}{\de \xi^a}\Big\rvert_{\xid} \vf^A \bigg\}_{\!a} \qquad \qquad a = 0, \ldots, 4
\eeq
are the tangent vectors to the cod-1 brane associated to the reference system $\xid$. This decomposition for the extrinsic curvature has a clear geometrical interpretation; the first two pieces are named ``orthogonal gradient'' and ``parallel gradient'' as they are non-zero when the bulk metric has non-zero derivative respectively in the directions orthogonal and parallel to the cod-1 brane, even when the cod-1 brane is not bent. The third piece is instead due to the bending, since it is non-zero when the brane is bent even if the bulk metric is constant.

Therefore, indicating with $\tmbfG$ the Einstein tensor built from the induced metric and with $\ti{K} = \textrm{tr} \, \tmbfK$ the trace of the extrinsic curvature, the equations of motion for our system are
\begin{align}
\label{Bulkeq}
\mbfG =& \,\, 0 \qquad \,\, (\textrm{bulk})\\[2mm]
\label{junctionconditionseq}
2 \Msf \Big( \tmbfK - \ti{K} \, \tmbfg \Big) + \Mft \, \tmbfG =& \,\, \tmbfT \qquad (\textrm{cod-1 brane})
\end{align}
where $\tmbfT$ is the (5D) energy-momentum tensor present inside the ribbon codimension-2 brane. The requirement that a $\mathbb{Z}_{2}$ symmetry across the cod-2 brane is assumed to hold inside the cod-1 brane is formalized asking that, when expressed in cod-1 GNC, the $\m\n$ and $\xi\xi$ components of the induced metric $\h{\mathbf{g}}$ (as well as of the curvature tensors $\h{\mathbf{G}}$ and $\h{\mathbf{K}}$ and of the energy-momentum tensor $\h{\mathbf{T}}$) are symmetric with respect to the reflection $\hxi \rightarrow - \hxi$, while the $\xi\m$ components are antisymmetric.

\section{Perturbations around pure tension solutions}
\label{Perturbations around pure tension solutions}

In the context of the nested branes with induced gravity, gravity regularization corresponds to the fact that the gravitational field on the ribbon brane remains finite when the thickness $l_2$ becomes smaller and smaller. Therefore, in this and the next section we study the gravitational field produced by a thick ribbon brane, and then consider the limit $l_2 \rightarrow 0^+$ in section \ref{ThinLimitCod2Brane}. As we already mentioned, we consider configurations when the gravitational field produced by the matter is weak, but we allow for a generic tension (vacuum energy) to be present on the cod-2 brane. In this section we derive the pure tension solutions, and develop a formalism to study perturbations at first order around the (background) pure tension configurations. We indicate with an overbar the quantities which correspond to the background configurations.

\subsection{Pure tension solutions}
\label{Pure tension solutions}

Let's consider localized source configurations such that, in cod-1 Gaussian Normal Coordinates, the (cod-1) energy-momentum tensor is of the form
\beq
\label{marina}
\bar{T}_{ab}(\hxi, \chd) = - \d_{a}^{\, \, \m} \,\, \d_{b}^{\, \, \n} \,\, f(\hxi) \, \bla \,\, \bar{g}_{\m\n}(\hxi, \chd)
\eeq
where $\bla > 0$ and the \emph{localizing function} $f$ vanishes for $\abs{\hxi} > l_2$ and satisfies
\beq
\label{Samatorza}
\int_{-}^{+} \! d\hxi \, f(\hxi) = 1
\eeq
This function can be considered to be a regularized version of the Dirac delta function, and to be compatible with the $\mathbb{Z}_2$ symmetry present inside the cod-1 brane it has to be even with respect to the reflection $\hxi \rightarrow - \hxi$. On every $\hxi$--constant (4D) hypersurface, the source configurations (\ref{marina}) have the form of pure tension, where the total tension $\bla$ is distributed according to the function $f$: we define ``thick pure tension source'' a source configuration which satisfy (\ref{marina}) and (\ref{Samatorza}), where $f$ descibes the internal structure of the ribbon cod-2 brane. When the cod-2 brane become thin and $f$ tends to a Dirac delta, the cod-2 energy momentum tensor tends to
\beq
\label{Ternova}
\bar{T}^{(4)}_{\m\n}(\chd) \rightarrow - \bla \,\, \bar{g}^{(4)}_{\m\n}(\chd)
\eeq
which is the energy-momentum tensor corresponding to a thin pure tension source (note that the minus sign is due to the choice of signature for the metric).

To find a solution to the equations of motion, we follow \cite{Dvali:2006if} (see also \cite{Gregory:2001xu,Gregory:2001dn}) and consider a geometrical ansatz which enjoys translational invariance in the 4D directions parallel to the (mathematical) cod-2 brane: we assume that $\mathcal{C}_2$ is placed at $\xi = 0$
\beq
\label{Robert}
\bar{\a}^{a}(\chi^{\cdot}) = \big( 0, \chi^{\m} \big)
\eeq
while the cod-1 brane has the following embedding
\beq
\label{Patty}
\bar{\varphi}^{A}(\xi^{\cdot}) = \big( Z(\xi), Y(\xi), \xi^{\m} \big)
\eeq
and the bulk metric is the 6D Minkowski metric
\beq
\label{George}
\bar{g}_{AB}(X^{\cdot}) = \eta_{AB}
\eeq
We assume furthermore that the function $Y(\xi)$ is a diffeomorphism, which in particular means that $Y^{\p}$ never vanishes\footnote{We indicate derivatives with respect to $\xi$ (or $\hxi$) with a prime $\dexi \equiv \phantom{i}^{\p}$.}. We can then use the gauge freedom to rescale the coordinate $\xi \rightarrow \hxi$ in such a way that $\bar{g}_{\xi\xi} (\hxi) = 1$, which implies 
\beq
\label{Kelly}
{Z^{\prime}}^{2}(\hxi) + {Y^{\prime}}^{2}(\hxi) = 1
\eeq
and therefore in full generality we can write $\Zp$ and $\Yp$ as
\begin{align}
\label{Jerry}
\Zp\big( \hxi \big) &= \sin S \big( \hxi \big)  & \Yp\big( \hxi \big) &= \cos S \big( \hxi \big) 
\end{align}
where $S$ is smooth since $\Zp$ and $\Yp$ are, and is called the \emph{slope function}. The name is justified by the fact that the slope of the embedding in the $y-z$ plane at the point $\hxi$ is given by
\beq
\label{Manelli}
\vartheta \big( \hxi \big) = \arctan \, \frac{\Zp\big( \hxi \big)}{\Yp\big( \hxi \big)} = S\big( \hxi \big)
\eeq
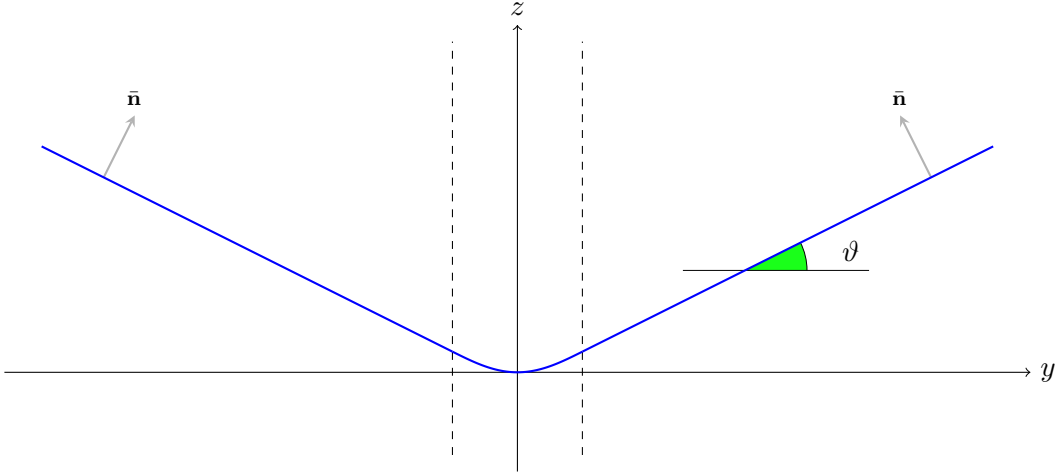
\begin{figure}[tbh!]
\centering
\begin{tikzpicture}[scale=1.1]
\draw[->] (-6.2,0) -- (6.2,0) node[right] {$y$};
\draw[->] (0,-1.2) -- (0,4.2) node[above] {$z$};
\draw[dashed] (0.78525,-1) -- (0.78525,4);
\draw[dashed] (-0.78525,-1) -- (-0.78525,4);
\fill[green!90!white] (2.75,1.232375) -- (3.5,1.232375) arc (0:27:0.75) -- (2.75,1.232375);
\draw[very thin] (2,1.232375) -- (4.25,1.232375) node[anchor=south east] {$\vartheta$};
\draw[very thin] (3.5,1.232375) arc (0:27:0.75);
\begin{scope}[>=stealth]
\draw[->,color=black!30!white,thick] (5,2.357375) -- (4.625,3.107375) node[above,color=black] {$_{\bar{\mathbf{n}}}$};
\draw[->,color=black!30!white,thick] (-5,2.357375) -- (-4.625,3.107375) node[above,color=black] {$_{\bar{\mathbf{n}}}$};
\end{scope}
\draw[color=blue,domain=-0.78525:0.78525,thick,smooth] plot (\x,{(1 - cos(2*\x r))/4});
\draw[color=blue,thick] (0.78525,0.25) -- (5.75,2.732375); 
\draw[color=blue,thick] (-0.78525,0.25) -- (-5.75,2.732375);
\end{tikzpicture}
\caption[The embedding of the cod-1 brane and its slope.]{The embedding of the cod-1 brane (thick, blue) and its slope. The vertical dashed lines indicate the boundaries of the cod-2 brane.}
\label{Slope function fig}
\end{figure}
(see figure \ref{Slope function fig}). Note that in the $(\hxi,\xim)$ reference system the metric induced on the cod-1 brane $\bar{g}_{ab}$ is the 5D Minkowski metric (in particular $(\hxi,\xim)$ is a cod-1 GN reference system) and that the metric induced on the cod-2 brane $\bar{g}^{_{(4)}}_{\m \n}$ is the 4D Minkowski metric. Using (\ref{Kelly}), the 6D 1-form normal to the cod-1 brane reads
\beq
\label{Hanayo}
\bar{n}_{M}(\hxi) = \varepsilon \, \big( Y^{\prime}(\hxi), - Z^{\prime}(\hxi), 0, 0, 0, 0 \big)
\eeq
and the only non-vanishing component of the extrinsic curvature of the cod-1 brane is
\beq
\bar{K}_{\xi \xi}(\hxi) = \varepsilon \, S^{\p}(\hxi)
\eeq
where $\varepsilon = \pm 1$ encodes the choice of the orientation of the cod-1 brane.

\subsubsection{Pure tension and def\mbox{}icit angle}

It is easy to see that the bulk equations of motion are identically satisfied, while the only components of the junction conditions which are not trivially satisfied are the $\m\n$ ones, which give
\beq
\label{Botswana2}
2 M_6^4 \, \bar{K}_{\xi \xi}(\hxi) = f(\hxi) \, \bla
\eeq
Note that, since the function $f(\hxi)$ is even, $S^{\p}(\hxi)$ has to be even as well; we choose to impose the condition $S(0) = 0$, so that $S(\hxi)$ is odd, which in turn implies that $\Zp \big( \hxi \big)$ is odd and $\Yp \big( \hxi \big)$ is even. The function $S(\hxi)$ is therefore determined by the Cauchy problem
\beq
\label{Lory}
\left\{
\begin{aligned}
S^{\p}(\hxi) &= \vep \, \frac{\bla}{2 M_6^4} \, f(\hxi) \\[2mm]
S(0) &= 0
\end{aligned}
\right.
\eeq
whose solution is
\beq
\label{satisfaction}
S(\hxi) = \vep \, \frac{\bla}{2 \Msf} \, \ep \big( \hxi \big)
\eeq
where we introduced the \emph{regulating function}
\beq
\ep \big( \hxi \big) \equiv \int_{0}^{\hxi} f (\z) \, d \z
\eeq
If $\bla \geq 2 \pi \Msf$, there exists a value $\hxi_{_{M}}$ (with $\abs{\hxi_{_{M}}} \leq l_2$) such that $Y^{\p}$ vanishes and changes sign at $\hxi = \hxi_{_{M}}$: in these cases, the cod-1 brane has a finite extent in the $\hxi$ direction. More precisely, in these cases the cod-1 and the cod-2 brane coincide and are both ribbon branes; to have a cod-1 brane which is infinite also in the $\hxi$ direction we have to impose $\bla < 2 \pi \Msf$, in which case $\Yp$ is a diffeomorphism $\mathbb{R} \rightarrow \mathbb{R}$. We conclude that there exists a maximum tension $\bla_M = 2 \pi \Msf$ for the thick pure tension solutions in the nested branes set-up to be well-defined. We assume henceforth that $\bla < \bla_M$; in this case, the functions $S(\hxi)$, $\Zp(\hxi)$ and $\Yp(\hxi)$ are constant for $\abs{\hxi} \geq l_2$. Imposing that $Z(0) = Y(0) = 0$, so that $Z \big( \hxi \big)$ is even and $Y \big( \hxi \big)$ is odd, for $\hxi \gtrless \pm l_2$ we have (see appendix \ref{Thin limit of the background})
\begin{align}
Z(\hxi) &= \Zp_+ \,\, \abs{\hxi} + Z_0 & Y(\hxi) &= \Yp_+ \,\, \hxi \pm Y_0 & \mathrm{for} \,\, \hxi &\gtrless \pm l_2
\end{align}
where $Z_0$ and $Y_0$ are integration constants and
\begin{align}
\label{Dom}
S_{+} &= \varepsilon \, \frac{\bla}{4 M_6^4} & Z^{\p}_{+} &= \sin \bigg( \varepsilon \, \frac{\bla}{4 M_6^4} \bigg) & Y^{\p}_{+} &=  \cos \bigg( \varepsilon \, \frac{\bla}{4 M_6^4} \bigg)
\end{align}

The slope of the embedding in the $y-z$ plane outside the thick cod-2 brane depends only on the total amount of tension $\bla$, and is independent of the internal structure of the thick cod-2 brane; furthermore, if we keep $\bla$ constant and perform a limit in which $l_2 \rightarrow 0^+$, $\Zp_+$ and $\Yp_+$ remain constant while $Z_0$ and $Y_0$ tend to zero (see appendix \ref{Thin limit of the background}). Therefore, if we restrict ourselves to configurations of the type (\ref{Ternova})--(\ref{George}), we can give a thin limit description to this set-up where the tension $\bla$ is perfectly localized on the mathematical cod-2 brane $\mcal{C}_2$ and the components of the embedding function read
\beq
\label{Iris}
\left\{
\begin{aligned}
Z(\hxi) &= \varepsilon \, \sin \Big( \frac{\bla}{4 M_6^4} \Big) \,\, \abs{\hxi} \\[2mm]
Y(\hxi) &= \cos \Big( \frac{\bla}{4 M_6^4} \Big) \,\, \hxi
\end{aligned}
\right.
\eeq
for every value of $\hxi$. At first sight, there are two solutions for each value of $\bla$, which correspond to the choices $\vep = +1$ and $\vep = -1$: however, it is not difficult to see that the two solutions $\vep = +1$, $\Sp >0$ and $\vep = -1$, $\Sp <0$ actually give rise to the same spacetime. For definiteness, henceforth we choose $\vep = +1$. From a geometrical point of view, the bulk correspondent to the solution defined by (\ref{Robert})-(\ref{George}), (\ref{Jerry}) and (\ref{satisfaction}) is the product of the 4D Minkowski space and of the region of the 2D Euclidean space which lies above the curve $\big( Y(\hxi) , Z(\hxi) \big)$ in the $y-z$ plane. The thin limit solution (\ref{Robert})-(\ref{George}) and (\ref{Iris}) then corresponds to a geometric configuration which is the product of the 4D Minkowski space and a two dimensional cone of deficit angle $\a = 4 \, \vth_+$, and using (\ref{Manelli}) and (\ref{Dom}) we obtain
\beq
\label{sandra}
\a = \frac{\bla}{M_6^4}
\eeq
Therefore, pure tension $\bla$ on a thin cod-2 brane produces a conical defect of deficit angle $\bla/M_6^4$ in the case of nested branes with induced gravity, generalizing the well-known result which holds for pure cod-2 branes. Note that when $\bla \to \bla_{M}^{-}$ the deficit angle tends to $2 \pi$, and the 2D cone tends to a degenerate cone (a half-line). The fact that the thick pure tension solutions are well-defined only if $\bla < \bla_M$ is reflected in the thin limit by the fact that the thin 6D nested branes configurations tend to a singular configuration when $\bla \to \bla_{M}^{-}$.

\subsection{Perturbations in the bulk-based approach}

We now turn to the study of perturbations. When studying metric theories of gravity, a very important step in solving the equations of motion is the choice of the reference system. In fact, choosing the gauge wisely is often crucial to be able to solve the equations explicitly. In cod-1 braneworld models, there exist two choices which are commonly used in the literature: the brane-based and the bulk-based approach. It is very important to understand which of the two approaches is best suited to our set-up, where the source on the cod-1 brane lies inside a ribbon whose thickness we want eventually to send to zero.

\subsubsection{Bulk-based and brane-based approach}

As a guide, we can look at the pure tension solutions, which have been derived in the previous section in the bulk-based approach. First of all, the global geometry of the solution is very transparent in this approach, since the deficit angle is directly connected to the slope of the embedding. In a brane-based approach, this information would be encoded indirectly in the form of the bulk metric. Secondly, when we send $l_2$ to zero, the normal vector (\ref{Hanayo}) becomes discontinuous at $\hxi = 0$ (i.e.~at the mathematical brane $\mcal{C}_2$, see appendix \ref{Thin limit of the background}) and the extrinsic curvature diverges there, while vanishing for $\hxi \neq 0$. In such a situation, we cannot impose (bulk) Gaussian Normal coordinates in a neighbourhood of $\mcal{C}_2$, since a necessary condition to do that is that the normal vector field is smooth \cite{CarrollBook}. Furthermore, even if we just ask that the cod-1 brane is straight, the bulk metric must be (at least) discontinuous in the bulk, because the discontinuity of the normal vector is now encoded entirely in the bulk metric (see \cite{deRham:2010rw}; of course, the singularity/discontinuity of the bulk metric is a pure coordinate artefact). This is true also for more general configurations: since the extrinsic curvature (\ref{ExtrCurvOG})--(\ref{ExtrCurvB}) is constructed from first derivatives of the bulk metric, if we fix the cod-1 brane embedding to be straight then we need a discontinuous or diverging bulk metric to generate an extrinsic curvature which diverges at one point. Therefore, a configuration where gravity is regularized is nevertheless reflected in the brane-based approach by the presence of discontinuities/divergences, with the difference that, if gravity is not regularized, then the discontinuities/divergences are not coordinate artefacts: to judge if gravity is regularized, we should perform a careful study of the behaviour of the geometry at the problematic points. 

On the other hand, this is not necessarily the case in a bulk-based approach: since the extrinsic curvature is built from \emph{second} derivatives of the embedding function, there exist a class of configurations (those where the embedding is cuspy and the bulk metric is smooth) where the extrinsic curvature diverges at $\mcal{C}_2$, while the bulk curvature tensors and the induced curvature tensors are finite. Therefore, the bulk based approach is the only choice where the fact that gravity is finite is reflected by the property that the configuration (metric and embedding) is \emph{continuous}. In our perturbative study, we would like to use an approach in which the properties of the bulk metric and cod-1 embedding reflects most clearly the fact that the intrinsic geometry diverges or not when the thin limit on the cod-2 brane is taken: therefore, we will use the bulk-based approach to perform our analysis, and we won't fix the embedding function.

\subsubsection{Perturbation of the geometry}

This in practice means that we leave both the bulk metric and the cod-1 embedding free to fluctuate, so we consider the perturbative decomposition
\begin{align}
g_{AB}(X^{\cdot}) &= \bar{g}_{AB}(X^{\cdot}) + h_{AB}(X^{\cdot}) \\[2mm]
\varphi^{A}(\xi^{\cdot}) &= \bar{\varphi}^{A}(\xi^{\cdot}) + \d\!\varphi^{A}(\xi^{\cdot})
\end{align}
while we decide to keep fixed the position of the cod-2 brane in the cod-1 coordinate system (\emph{i.e.} it is still located at $\xi = 0$). In addition, we still use the 4D coordinates of the cod-1 brane to parametrize the cod-2 brane, so the embedding of the cod-2 brane reads
\beq
\label{Joan}
\tilde{\a}^{a}(\chi^{\cdot}) = \bar{\a}^{a}(\chi^{\cdot}) = \big( 0, \chi^{\m} \big)
\eeq 
also at perturbative level. We def\mbox{}ine the perturbations of the metric induced on the cod-1 brane as follows
\beq
\label{Caribbean}
\tilde{h}_{ab}(\xi^{\cdot}) \equiv  \tilde{g}_{ab}(\xi^{\cdot}) - \bar{g}_{ab}(\xi^{\cdot})
\eeq
and analogously we define the perturbation of the metric induced on the cod-2 brane as
\beq
h^{(4)}_{\m \n}(\chi^{\cdot}) \equiv g^{(4)}_{\m \n}(\chi^{\cdot}) - \bar{g}^{(4)}_{\m \n}(\chi^{\cdot})
\eeq
It follows that, since the embedding of the cod-2 brane in the cod-1 brane is trivial, the perturbation of the metric induced on the cod-2 brane is linked to the perturbation of the metric induced on the cod-1 brane by the simple relation
\beq
h^{(4)}_{\m \n}(\chi^{\cdot}) = \tilde{h}_{\m \n}(0,\chd)
\eeq

When considering quantities which are decomposed in a background and a perturbation part, we use the convention that indices on the perturbation part of every quantity (and on the background part as well) are lowered/raised with the background metric. For example, considering the parallel vectors $\mathbf{v}_{(a)}$, we consider the perturbative decomposition
\beq
v^{A}_{(a)} = \bv^{A}_{(a)} + \dv^{A}_{(a)}
\eeq
where
\begin{align}
\bv^{A}_{(a)} &\equiv \frac{\de \bvf^A}{\de \xi^a} & \dv^{A}_{(a)} &\equiv \frac{\de \, \dvf^A}{\de \xi^a}
\end{align}
and the index-lowered background and perturbation parts read
\begin{align}
\bv_{A}^{(a)} &\equiv \e_{AB} \, \bv^{B}_{(a)} & \dv_{A}^{(a)} &\equiv \e_{AB} \, \dv^{B}_{(a)}
\end{align}
As a matter of fact, the vectors $\mathbf{v}_{(\m)}$ will not play a crucial role in the analysis, so in the following we will indicate $\bv^{A}_{(\xi)}$ and $\dv^{A}_{(\xi)}$ simply with $\bv^{A}$ and $\dv^{A}$. The 2D indices $i$, $j$ and $k$, which run on the extra dimensions $z$ and $y$, are raised/lowered with the identity matrix, so we have for example
\begin{align}
\bv_i = \bvf_i^{\p} &\equiv \d_{ij} \, \bvf^{j \, \p} & \bn^{i} &\equiv \d^{ij} \, \bn_{j}
\end{align}

\subsubsection{Perturbation of the source}

Concerning the source term, we consider a perturbed energy-momentum tensor which in cod-1 GNC is of the form
\beq
\label{Ipanema}
\hat{T}_{ab}(\hxi, \chd) = - \d_{a}^{\, \, \m} \,\, \d_{b}^{\, \, \n} \, f(\hxi) \, \big( \bla + \dla \big) \,\h{g}_{\m\n}(\hxi, \chd) + \h{\mcalT}_{ab}(\hxi, \chd)
\eeq
where $\h{\mcalT}_{ab}$ is the energy-momentum of the matter present inside the (thick) cod-2 brane and satisfies the relations (\ref{angelica}). Note that we perturb both the matter content and the tension ($\dla$) of the cod-2 brane. At linear order, the equation (\ref{Ipanema}) reads
\beq
\hat{T}_{ab} = \bar{T}_{ab} - \d_{a}^{\, \, \m} \,\, \d_{b}^{\, \, \n} \, f(\hxi) \, \bla \, \h{h}_{\m\n} + \d \hat{T}_{ab}
\eeq
where $\bar{T}_{ab}$ is the (thick) background pure tension source term, $- \d_{a}^{\, \, \m} \,\, \d_{b}^{\, \, \n} \, f(\hxi) \, \bla \, \h{h}_{\m\n}$ is a perturbation term coming from the background tension and
\beq
\d \hat{T}_{ab}(\hxi, \chd) = - \d_{a}^{\, \, \m} \,\, \d_{b}^{\, \, \n} \, f(\hxi) \,\, \dla \,\, \e_{\m\n} + \h{\mcalT}_{ab}(\hxi, \chd)
\eeq
is the perturbation term due to the tension perturbation and to the matter. Note that in principle $\bar{T}_{ab}$ and $\d \hat{T}_{ab}$ may be characterized by different localization lengths $l_2$ and $l^{\p}_2$: we ask that they are of the same order of magnitude $l_2 \sim l^{\p}_2$, and in the following for simplicity we indicate with $l_2$ the biggest between $l_2$ and $l^{\p}_2$. In particular, in the following we implicitly assume this convention when we use the notation defined in the equations (\ref{strange equation1})-(\ref{strange equation2}). It is important to notice that, in principle, the presence of the matter may alter the distribution of the tension inside the thick cod-2 brane, as a consequence of generalized Casimir effects, and that the localizing function may be dependent on the amount of tension. We assume that such effects are not present, and in particular the form of the background solution and the form of $\h{\mcalT}_{ab}$ are independent in our analysis. In analogy to what we did above, we define the matter cod-2 energy-momentum tensor as follows
\beq
\label{surprise}
\mcalT^{(4)}_{\m\n}(\chd) \equiv \int_{-}^{+} \! d \hxi \,\, \h{\mcalT}_{\m\n}(\hxi, \chd)
\eeq

\subsection{The scalar-vector-tensor decomposition}

We motivated above that, when considering nested branes with induced gravity set-ups, it is convenient to use a bulk based approach where the embedding of the cod-1 brane is free to fluctuate. 
However, fixing the gauge is in general very useful because it permits to get rid of pure gauge perturbation modes and work only with physical perturbation modes. Therefore, if we want to follow the route sketched above, we need a way to recognize and separate pure gauge modes from physical modes without fixing any gauge in the bulk.

We achieve this by performing a scalar-vector-tensor decomposition with respect to the 4D Lorentz group (acting on the coordinates $\xm$ in the bulk, on the coordinates $\xim$ on the cod-1 brane, and on the coordinates $\chm$ on the cod-2 brane $\mcal{C}_2$), and by working with gauge invariant variables. Note that, to perform this decomposition, we have to impose 4D boundary conditions which render the D'Alembert operator $\boxf$ invertible; we implicitly assume this in the following. Most importantly, the convenience of using this decomposition is also that, at linear order, the three sectors (tensor, vector and scalar) decouple, and so the equations of motion in each sector may be simpler to solve than the complete equations. In particular, we shall see below that the perturbation modes of the embedding of the cod-1 brane (\emph{bending modes}) play a role only in the scalar sector, where they are crucial for the regularization of gravity: so the decomposition also helps us to focus on the subtleties involved with the bending modes and on the mechanism we want to study.

\subsubsection{Metric and bending decomposition}

We consider the following decomposition of the bulk metric perturbation in TT-tensor, T-vector and scalar parts
\begin{align}
h_{\mu \nu} &= \mscr{H}_{\mu \nu} + \partial_{(\mu} V_{\nu)} + \e_{\mu \nu} \,\pi + \partial_{\mu}\partial_{\nu} \vp \\[2mm]
h_{z \mu} &= A_{z \mu} + \partial_{\mu} \s_{z} \\[2mm]
h_{y \mu} &= A_{y \mu} + \partial_{\mu} \s_{y} \\[2mm]
h_{yy} &= \psi \\[2mm]
h_{zy} &= \r \\[2mm]
h_{zz} &= \o
\end{align}
where all the quantities are functions of the bulk coordinates $X^{\cdot}$ and we use the notation $\dem \equiv \de/\de x^\m$. In particular, $\mscr{H}_{\mu \nu}$ is a transverse-traceless symmetric tensor while $V_\mu$ , $A_{z \mu}$ and $A_{y \mu}$ are transverse 1-forms, and $\o$, $\r$, $\psi$, $\s_{z}$, $\s_{y}$, $\pi$ and $\vp$ are scalars. Regarding the scalar part of the $\m\n$ components, we call \emph{trace part} the scalar field which multiplies the Minkowski metric (in this case $\pi$), while we call \emph{derivative part} the scalar field derivated twice with respect to the 4D coordinates (in this case $\vp$). Concerning the codimension-1 brane embedding, the bending modes $\dvf^z$ and $\dvf^y$ are scalars; for the 4D components, we consider the decomposition
\beq
\dvf_{\m} = \dvf^{_{T}}_{\m} + \dexim \dvf_4
\eeq
where $\dvf_4$ is a scalar and $\dvf^{_{T}}_{\m}$ is a transverse vector, and $\dexim \equiv \de/\de \xi^\m$.

Similarly, we perform the scalar-vector-tensor decomposition of the cod-1 induced metric with respect to the 4D coordinates $\xim$. Note that, if we consider only the 4D coordinates $x^\m$ in the bulk, the background embedding of the cod-1 brane into the bulk is trivial; this implies that the scalar/vector/tensor sector of the cod-1 induced metric is constructed only with the corresponding sector in the bulk evaluated on the brane (and with the corresponding sector of the bending modes). In other words, passing from the bulk metric to the induced metric does not couple the different sectors. In turn, this is true also when we pass form the cod-1 brane to the cod-2 brane, where we consider a scalar-vector-tensor decomposition with respect to the coordinates $\chm$. To avoid a cumbersome notation, we use henceforth the convention that the evaluation on the cod-1 brane of a bulk quantity is indicated with a tilde, so for example
\begin{align}
\ti{\mscr{H}}_{\m\n}(\xid) &\equiv \mscr{H}_{\m\n}(\Xd) \evb & \ti{\pi}(\xid) &\equiv \pi(\Xd) \evb
\end{align}
and similar definitions hold for the other bulk fields. Likewise, we use the convention that the evaluation on the (mathematical) cod-2 brane of a bulk quantity is indicated with a superscript $\!\phantom{i}^{_{(4)}}$, so for example
\begin{align}
\mscr{H}^{(4)}_{\m\n}(\chd) &\equiv \mscr{H}_{\m\n}(\Xd) \Big\rvert_{X^{\cdot} = \bar{\varphi}^{\cdot}(0,\chd)} & \pi^{(4)}(\chd) &\equiv \pi(\Xd) \Big\rvert_{X^{\cdot} = \bar{\varphi}^{\cdot}(0,\chd)}
\end{align}

\subsubsection{Source decomposition}
\label{source decomposition}

Regarding the matter cod-1 energy-momentum tensor, we consider the following decomposition 
\begin{align}
\ti{\mcalT}_{\m\n} &= \ti{\mscrT}_{\m\n} + \de_{\xi^{( \m}} \, \ti{B}_{\n)} + \deximxin \, \ti{\mcalT}_{de} + \eta_{\m\n} \, \ti{\mcalT}_{tr}  \\[2mm]
\ti{\mcalT}_{\xi \mu} &= \ti{D}_{\mu} + \partial_{\xim} \tt
\end{align}
where the symmetric tensor $\ti{\mscrT}_{\m\n}$ is transverse and traceless while $\ti{B}_{\m}$ and $\ti{D}_{\m}$ are transverse 1-forms and $\ti{\mcalT}_{\xi \xi}$, $\ti{\mcalT}_{tr}$, $\ti{\mcalT}_{de}$ are scalars. We consider also the scalar-vector-tensor decomposition of the cod-2 energy-momentum tensor with respect to the coordinates $\chd$
\beq
\mcalT^{(4)}_{\m\n} = \mscrT_{\m\n}^{(4)} + \de_{\ch^{( \m}} \, B^{(4)}_{\n)} + \dechmchn \, \mcalT^{(4)}_{de} + \eta_{\m\n} \, \mcalT^{(4)}_{tr}
\eeq
where $\mscrT_{\m\n}^{_{(4)}}$, $B^{_{(4)}}_{\m}$, $\mcalT^{_{(4)}}_{de}$ and $\mcalT^{_{(4)}}_{tr}$ are respectively the pillbox integration of $\h{\mscrT}_{\m\n}$, $\h{B}_{\m}$, $\h{\mcalT}_{de}$ and $\h{\mcalT}_{tr}$ while the pillbox integration of $\h{D}_{\mu}$, $\h{\t}$ and $\h{\mcalT}_{\xi\xi}$ vanish as a consequence of (\ref{angelica}). Note that taking the divergence of the cod-2 energy momentum tensor we get
\beq
\de^{\xi^{\n}} \mcalT^{(4)}_{\m\n} = \boxf B^{(4)}_{\m} + \dexim \Big( \boxf \mcalT^{(4)}_{de} + \mcalT^{(4)}_{tr} \Big)
\eeq
while taking the double divergence of the same tensor we get
\beq
\de^{\xim} \de^{\xi^{\n}} \mcalT^{(4)}_{\m\n} = \boxf \Big( \boxf \mcalT^{(4)}_{de} + \mcalT^{(4)}_{tr} \Big)
\eeq
Since we assume that the cod-2 energy momentum tensor is covariantly conserved, the invertibility of the operator $\boxf$ implies
\begin{align}
\boxf \mcalT^{(4)}_{de} + \mcalT^{(4)}_{tr} &= 0 \label{nottoobad1} \\[2mm]
B^{(4)}_{\m} & = 0 \label{nottoobad2}
\end{align}
which in particular means that the vector sector is not sourced at first order in perturbations when the cod-2 brane is thin. This result corresponds to the well-known result that, when considering first order perturbations around the 4D Minkowski spacetime, the interaction term in the action between the vector mode and the energy-momentum tensor is a total derivative if the latter is covariantly conserved.

\subsubsection{Bulk gauge invariant variables}

In order to identify the physical perturbation modes and the pure gauge perturbation modes, and to work only with the former modes, we recast the theory in terms of gauge invariant (g.i.) variables. Considering an infinitesimal change of coordinates in the bulk
\beq
X^{\p \, A} = X^A - \La^A(\Xd)
\eeq
the metric tensor transforms as follows
\beq
h^{\p}_{MN}(\Xd) = h_{MN}(\Xd) + \de_{X^{(M}} \,\, \La_{N)}(\Xd)
\eeq
and the bending modes transform as
\beq
\dvf^{\p \, A}(\xid) = \dvf^{A}(\xid) - \ti{\La}^A(\xid)
\eeq
where we defined $\ti{\La}^{L} \equiv \La^{L}(\bvf^{\cdot}(\xid))$ and it is intended that a prime here does not denote a derivative but just the fact that the quantities are expressed in the new coordinate system. In the context of the 4D scalar-vector-tensor decomposition, we can decompose the gauge parameter as follows: defining the index lowered gauge parameter as $\La_{N} = \eta_{NL} \, \La^{L}$, we have that $\La^{z} = \La_{z}$ and $\La^{y} = \La_{y}$ transform as scalars, while the 4D part decompose as
\beq
\La_{\m} = \La^{_{T}}_{\m} + \dem \La_{4}
\eeq
where $\La_{4}$ is a scalar while $\La^{_{T}}_{\m}$ is a transverse 1-form.

It is straightforward to check that the tensor part of the bulk metric perturbations is gauge invariant
\beq
h^{gi}_{\mu \nu} = \mscr{H}_{\mu \nu}
\eeq
while the gauge invariant variables for the vector part of the bulk metric perturbations are
\begin{align}
\mcal{A}_{z \m} = A_{z\m}^{gi} &\equiv A_{z \m} - \de_z \, V_{\m} \\[2mm]
\mcal{A}_{y \m} = A_{y\m}^{gi} &\equiv A_{y \m} - \de_y \, V_{\m}
\end{align}
and the gauge invariant variables for the scalar part of the bulk metric perturbations are
\begin{align}
\pi^{gi} &\equiv \pi \\[2mm]
h_{ij}^{gi} &\equiv h_{ij} - \ded{(i} \, \s_{j)} + \de_{i} \de_{j} \, \varpi
\end{align}
where $h_{ij}^{_{gi}}$ synthetically indicates the variables $\o^{_{gi}}$, $\r^{_{gi}}$ and $\ps^{_{gi}}$.

\subsubsection{Brane gauge invariant variables}

The Einstein equations in the bulk can be written in terms of gauge invariant variables alone, as it is easy to check. This is a consequence of the covariance of the Einstein equations, and of the fact that we can choose a gauge (the one where $V_\m = \vp = \s_z = \s_y = 0$) in which the g.i.~variables are equal to the non-zero components of the metric perturbations. We may wonder if also the perturbation of the embedding of the cod-1 brane can be described in a gauge invariant way, since at first sight the presence of the brane seems to break the diffeomorphism invariance in the bulk. This invariance is indeed broken if we describe the brane as a fixed boundary, or if we constrain in any way its position. However, if we allow the brane to fluctuate freely, the presence of the bending modes in the equations of motion actually enforces the diffeomorphism invariance in the bulk. This can be understood also from the fixed-boundary point of view: in that case, the broken diffeomorphism invariance in the bulk can be reinstated using St\"uckelberg fields, which however are physically interpreted as bending modes \cite{LutyPorratiRattazzi,NicolisRattazzi}.

In fact, we can give a gauge invariant description of the fluctuation in the brane position by considering the following gauge invariant versions of the bending modes
\begin{align}
\dvf^{i}_{gi} &\equiv \dvf^{i} + \Big[ \s_i - \half \, \de_i \vp \Big]\evb \\[1mm]
\dvf^{gi}_{4} &\equiv \dvf_{4} + \half \, \tvp \\[3mm]
\dvf_{\m}^{_{(T)}} &\equiv \dvf_{\m}^{_{T}} + \ti{V}_{\m}
\end{align}
where $\dvf^{z}_{^{gi}}$, $\dvf^{y}_{^{gi}}$ and $\dvf^{_{gi}}_{4}$ belong to the scalar sector while $\dvf_{\m}^{_{(T)}}$ belongs to the vector sector. Henceforth we refer to these variables as the brane-gauge invariant variables. Note that, in the gauge $V_\m = \vp = \s_z = \s_y = 0$, they coincide with the bending modes. 

The possibility to describe in a gauge invariant way both the perturbations of the bulk metric and the perturbations of the brane embedding, permits in the nested brane set-up to study the perturbations around the pure tension solutions in a purely gauge-invariant way.

\section{The equations of motion for the perturbations}
\label{The equations of motion for the perturbations}

\subsection{Bulk equations of motion and master variables}
\label{Bulk equations of motion and master variables}

The Einstein equations in the bulk (\ref{Bulkeq}) decompose into three groups of equations respectively for the tensor, vector and scalar part of the metric perturbations. The equation for the tensor sector is
\beq
\label{Bulkeqtensor}
\boxs \, \mscr{H}_{\m\n} = 0
\eeq
while the equations for the vector sector are
\begin{align}
\de_i \de^j \mcal{A}_{j\m} - \boxs \, \mcal{A}_{i\m} &= 0 \\[2mm]
\de^i \de_{(\m} \, \mcal{A}_{i|\n)} &= 0 \label{afterall}
\end{align}
Note that taking the 4D divergence $\de^{\n}$ of the equation (\ref{afterall}) we get
\beq
\boxf \, \de^i \mcal{A}_{i\m} = 0
\eeq
and applying the inverse of the 4D D'Alembert operator (which we indicate with $1/\boxf$) we obtain
\beq
\de^i \mcal{A}_{i\m} = 0
\eeq
Therefore, the bulk equations of motion for the vector sector are
\begin{align}
\boxs \, \mcal{A}_{i\m} &= 0 \label{Gina e Enzo} \\[2mm]
\de^i \mcal{A}_{i\m} &= 0 \label{Enzo e Gina}
\end{align}
where the equation (\ref{Gina e Enzo}) synthetically expresses the two equations for $\mcalA_{z\m}$ and $\mcalA_{y\m}$.

Regarding the scalar sector, in principle every component of the bulk Einstein equations produces an equation for the scalar sector. However, the Bianchi identity
\beq
\nabla_{M} \, G^{M}_{\,\, N} \simeq \de_{M} \, G^{M}_{\,\, N} = 0
\eeq
links together (in a differential way) the components of the Einstein tensor (independently of the fact that the metric solves the Einstein equations or not); this implies that it is sufficient to impose that $G^{z}_{\,\, z} = G^{y}_{\,\, z} = G^{y}_{\,\, y} = 0$ and that the trace part of the bulk Ricci tensor vanishes, to guarantee that the other components of the Einstein equations are satisfied. Therefore the bulk equations for the scalar sector read
\begin{align}
& \boxf \, h^{gi}_{ij} + 3 \, \d_{ij} \, \boxf \pi + 4 \, \de_{i} \de_{j} \, \pi = 0 \label{desiree} \\[3mm]
& \Box_6 \, \pi = 0 \label{Marcus} 
\end{align}
where the equation (\ref{desiree}) synthetically expresses the equations for $ij = zz$, $zy$ and $yy$.

\subsubsection{Scalar master variables}
\label{Scalar master variables}

The analysis can be greatly simplified by expressing the equations in terms of master variables \cite{Mukohyama:2000ui} (see also \cite{Kodama:2000fa}). Considering the scalar sector, a closer look to the bulk equations (\ref{desiree})--(\ref{Marcus}) reveals that the field $\pi$ obeys a decoupled equation; furthermore, the invertibility of the operator $\boxf$ implies that the other gauge invariant variables are completely determined in terms of $\pi$ by the equation (\ref{desiree}), since we have
\beq
\label{desiree2}
h^{gi}_{ij} = - 3 \, \d_{ij} \, \pi - \frac{4}{\boxf} \, \de_{i} \de_{j} \, \pi
\eeq
Moreover, evaluating the equation above on the cod-1 and cod-2 brane we can express the variables $\ti{h}^{_{gi}}_{ij}$ and $h^{_{(4)}}_{ij}$ in terms of $\ti{\pi}$ and $\pi^{_{(4)}}$: therefore, for the scalar sector we can actually work with only one metric variable, the master variable $\pi$.

Concerning the brane-gauge invariant variables, as we mentioned above the variable $\dvf^{_{gi}}_{4}$ does not appear in $\d\tilde{G}_{ab}$ and $\d \!\tilde{K}_{ab}$, so it does not appear at all in the equations of motion. Regarding the two remaining bending gauge invariant variables $\dvfzgi$ and $\dvfygi$, it is customary to describe the perturbations of the brane embedding by projecting the bending mode in the normal direction and in the parallel direction to the brane; we define therefore the normal component of the bending $\dvfn$ and the parallel component $\dvfp$
\begin{align}
\label{whoknows}
\dvfn &\equiv \bn_i \, \dvf_{gi}^{i} & \dvfp &\equiv \bv_i \, \dvf_{gi}^i
\end{align}
in terms of which the brane gauge invariant variables read
\beq
\dvf^{i}_{gi} = \dvfn \, \bn^{i} + \dvfp \, \bv^i
\eeq
It is important to keep in mind that, despite the normal and parallel components of the bending have an intuitive geometrical meaning when the normal vector is smooth, they are not well defined when the normal vector is discontinuous, while the $z$ and $y$ components of the bending remain well-defined.

A compelling reason for working with the normal and parallel components of the bending is that, when considering perturbations around a straight brane, the parallel component is a pure gauge mode since it does not appear in the equations of motion, and so the perturbation of the brane position can be described by only one master variable, $\dvfn$ . On the other hand, when the background configuration of the brane is bent, the perturbation of the extrinsic curvature contains both the normal component $\dvfn$ and the parallel component $\dvfp$ (we discuss this point in the appendix \ref{IndMetrandCurvTensApp}). Therefore, when the cod-2 brane is thick, we need to take explicitly into account $\dvfp$ as well and to work with two brane-master variables. However, this is not necessarily the case when we consider the thin limit of the cod-2 brane. In fact, in this limit the brane is straight everywhere apart from $\hxi = 0$ where the embedding is not derivable (see appendix \ref{Thin limit of the background}), and so $\dvfp$ disappears from the junction conditions for $\hxi \neq 0$. If we are able to express the junction conditions at the cod-2 brane entirely in terms of $\dvfn$ and $\pi$, the whole system (in the scalar sector) is described by two master variables: the ``metric'' master variable $\pi$, and the ``bending'' master variable $\dvfn$. This is indeed what happens in our set-up, as we shall see.

\subsubsection{Vector master variable}
\label{Vector master variable}

Regarding the vector sector, the vector bending mode $\dvf_{\m}^{_{(T)}}$ at first order does not contribute to the cod-1 Ricci and Einstein tensor and to the extrinsic curvature, as can be checked in appendix \ref{IndMetrandCurvTensApp}. Therefore the equations of motion for the vector sector involve only ``metric'' gauge invariant variables. Moreover, also in the vector sector we can introduce master variables. In fact, as it is shown in \cite{Mukohyama:2000ui}, the condition (\ref{Enzo e Gina}) implies that the gauge invariant variables $\mcalA_{z\m}$ and $\mcalA_{y\m}$ can be derived from a master variable $\Phi_{\m}$ as follows
\beq
\label{vectormasterdef}
\mcalA_{i\m} = \ep_{i}^{\,\, j} \, \de_j \, \Phi_{\m}
\eeq
where $\ep_{ij}$ is the totally antisymmetric symbol of unit weight (Levi-Civita symbol), and we choose the sign convention $\ep_{yz} = +1$. Note that the variable $\Phi_{\m}$ is \emph{not} unique, since it can be freely redefined by adding an arbitrary function of the 4D coordinates $\xm$. It is easy to check that with our convention we have
\begin{align}
\label{formula pero}
\bv^i \, \ep_{i}^{\,\, j}& = \bn^j & \bn^i \, \ep_{i}^{\,\, j}& = - \bv^j
\end{align}
and therefore
\begin{align}
\label{formula pero2}
\bv^i \, \h{\mcalA}_{i\m}& = \de_{\bar{\mathbf{n}}}\evbhsi \, \Phi_{\m} & \bn^i \, \h{\mcalA}_{i\m}& = - \de_{\bar{\mathbf{v}}}\evbhsi \, \Phi_{\m} = - \h{\Phi}_{\m}^{\p}
\end{align}

\subsection{The codimension-1 junction conditions}
\label{The codimension-1 junction conditions}

In analogy with the definition of the perturbation of the induced metric (\ref{Caribbean}), we define the perturbation of the cod-1 Einstein tensor and of the cod-1 extrinsic curvature as follows
\begin{align}
\d\ti{G}_{ab}(\xi) &\equiv \ti{G}_{ab}(\xi) - \bar{G}_{ab}(\xi) = \ti{G}_{ab}(\xi) \\[2mm]
\d\!\ti{K}_{ab}(\xi) &\equiv \ti{K}_{ab}(\xi) - \bar{K}_{ab}(\xi)
\end{align}
To study the junction conditions at perturbative level, we use henceforth the cod-1 Gaussian Normal Coordinates which we introduced in section \ref{The geometric set-up}. It may seem strange that, having paid attention to work in a gauge-invariant way in the bulk to avoid fictitious singularities produced by the bulk GNC, we choose now to impose GNC inside the cod-1 brane. To justify this, note first of all that it is always possible (at least locally) to impose Gaussian Normal Coordinates in each of the two domains which constitute the cod-1 brane and whose common boundary is the mathematical brane $\mcal{C}_2$. We can then join in a smooth way the two ``partial'' reference systems, and obtain a cod-1 Gaussian Normal Coordinates system (at least) in a neighbourhood of the cod-2 brane. This procedure works perfectly also in the thin limit, since the fact that the embedding becomes cuspy (from the bulk point of view) does not have any influence on our ability to impose the Gaussian Normal Coordinates in each of the two parts, and to join them smoothly (from a brane point of view). Therefore, the use of Gaussian Normal Coordinates inside the cod-1 brane is justified in our set-up.

\subsubsection{The perturbative junction conditions}

If we insert the perturbative decomposition of the metric and embedding around the background solutions in the junction conditions (\ref{junctionconditionseq}), we obtain the background junction conditions (which we disregard) plus a perturbation piece. The latter contains, in the left hand side, a term
\beq
- 2 M_6^4 \, \bar{K}_{\xi\xi} \, \h{h}_{ab}
\eeq
which cancels (as a consequence of the background relation (\ref{Botswana2}) and of the cod-1 GNC) the source term 
\beq
- \d_{a}^{\, \, \m} \,\, \d_{b}^{\, \, \n} \, f(\hxi) \, \bla \, \h{h}_{\m\n}
\eeq
which arises when we perturb the metric which multiplies the unperturbed tension. Disregarding these two terms as well, and noting that $\h{h}^{cd} \, \bar{K}_{cd}$ vanishes because of the cod-1 GNC, the perturbation of the junction conditions reads
\begin{align}
- 2 \, M_6^4 \, \e^{\m\n} \d \! \hat{K}_{\m\n} + M_5^3 \, \d \hat{G}_{\xi\xi} &= \d \hat{T}_{\xi\xi} \label{formula interessante1} \\[3mm]
2 M_6^4 \, \d \! \hat{K}_{\m \xi} + M_5^3 \, \d \hat{G}_{\xi \m} &= \d \hat{T}_{\m\xi} \label{formula interessante2} \\[3mm]
2 M_6^4 \Big[ \d \! \hat{K}_{\m\n} - \e_{\m\n} \Big( \d\!\h{K}_{\xi\xi} + \e^{\a\b} \d \! \h{K}_{\a\b} \Big) \Big] + M_5^3 \, \d \hat{G}_{\m\n} &= \d \hat{T}_{\m\n} \label{formula interessante3}
\end{align}
where all the quantities are functions of the cod-1 Gaussian Normal Coordinates $(\hat{\xi}, \chd)$. From the point of view of the scalar-vector-tensor decomposition, the equations (\ref{formula interessante1})--(\ref{formula interessante3}) decompose in separate junction conditions for each sector. In particular, the fields belonging to the tensor sector are present only in the $\m\n$ equation (\ref{formula interessante3}), while the fields belonging to the vector sector are present in the $\m\xi$ and $\m\n$ equations (\ref{formula interessante2}) and (\ref{formula interessante3}). As we mentioned above, the 4D components of the bending modes $\hdvf_{\m}$ do not contribute to the perturbation of the extrinsic curvature $\d \! \h{K}_{ab}$ and of the Einstein tensor $\d \! \h{G}_{ab}$; furthermore, the $z$ and $y$ components of the bending $\hdvf^{i}$ contribute only to the scalar sector of $\d \! \h{K}_{ab}$ and of $\d \! \h{G}_{ab}$ (see appendix \ref{IndMetrandCurvTensApp}). Therefore, the bending modes appear only in the scalar sector of the perturbative junction conditions (\ref{formula interessante1})--(\ref{formula interessante3}), and more in general of the equations of motion.

Using the formulae of appendix \ref{IndMetrandCurvTensApp}, the tensor sector of the perturbation of the cod-1 junction conditions reads explicitly

\beq
\label{maguardaunpo}
2 \Msf \, \de_{\bar{\mathbf{n}}}\evbhsi \, \mscr{H}_{\m\n} + \Mft \,\, \boxfi \, \h{\mscr{H}}_{\m\n} = - 2 \, \h{\mscrT}_{\m\n}
\eeq

\noi while the vector sector of the perturbation of the cod-1 junction conditions reads

\begin{align}
\label{avventura1}
2 M_6^4 \, \Big( \bv^i \, \de_{\bar{\mathbf{n}}}\evbhsi \, \mcalA_{i\m} - \bn^i \, \h{\mcalA}_{i\m}^{\p} \Big) + M_5^3 \, \boxf \, \bv^i \, \h{\mcalA}_{i\m} &= - 2 \, \h{D}_\m \\[3mm]
2 M_6^4 \, \de_{\ch^{( \m}} \, \bn^i \, \h{\mcalA}_{i|\n)} + M_5^3 \, \de_{\ch^{( \m}} \, \de_{\hxi} \Big( \bv^i \, \h{\mcalA}_{i|\n)} \Big) &= 2 \, \de_{\ch^{(\m}} \, \h{B}_{\n)}\label{avventura2}
\end{align}
Note that, taking the 4D divergence of the equation (\ref{avventura2}) and applying the operator $1/\boxf$, we get

\beq
\label{Sara}
2 M_6^4 \,\, \bn^i \h{\mcalA}_{i\m} + M_5^3 \, \de_{\hxi} \, \Big( \bv^i \h{\mcalA}_{i\m} \Big) = 2 \, \h{B}_{\m}
\eeq

Finally, considering the scalar sector, the $\xi\xi$ and the the $\xi\m$ components of the perturbation of the junction conditions eq.~(\ref{formula interessante1}) and eq.~(\ref{formula interessante2}) read respectively

\begin{align}
\label{xixijunctioncondgi}
2 \, M_6^4 \, \Big( 2 \, \de_{\bar{\mathbf{n}}}\evbhsi \, \pi - \boxf \, \bn_i \, \hdvf^{i}_{gi} \Big) + \frac{3}{2} \, M_5^3 \, \boxf \, \hp &= \h{\mcalT}_{\xi \xi} \\[3mm]
\label{ximujunctioncondgi}
2 \, M_6^4 \, \dechm \Big( \bn^i \bv^{j} \, \h{h}^{gi}_{ij} + 2 \, \bn_i \, \hdvf^{i \, \p}_{gi} \, \Big) - 3 \, M_5^3 \, \dechm \, \hp^{\p} &= 2 \, \dechm \, \hat{\t}
\end{align}
and, taking the 4D divergence of the equation (\ref{ximujunctioncondgi}) and applying the operator $1/\boxf$, we get
\beq
\label{ximujunctioncondgibis}
2 \, M_6^4 \, \Big( \bn^i \bv^{j} \, \h{h}^{gi}_{ij} + 2 \, \bn_i \, \hdvf^{i \, \p}_{gi} \, \Big) - 3 \, M_5^3 \, \, \hp^{\p} = 2 \, \hat{\t}
\eeq
Regarding the (scalar sector of the) $\m\n$ components of the perturbation of the junction conditions eq.~(\ref{formula interessante3}), the derivative part reads

\beq
\label{dermunujunctioncondgi}
2 M_6^4 \, \dechmchn \, \bn_i \, \hdvf^{i}_{gi} + M_5^3 \, \dechmchn \bigg( \! - \frac{1}{2} \, \bv^{i} \bv^{j} \, \h{h}^{gi}_{ij} - \hp + \bv_{i}^{\p} \, \hdvf^{i}_{gi} \bigg) = \dechmchn \hat{\mcal{T}}_{de}
\eeq
while the trace part reads
\begin{multline}
2 M_6^4 \, \bigg( \frac{3}{2} \, \de_{\bar{\mathbf{n}}}\evbhsi \, \pi + \frac{1}{2} \, \bv^{i} \bv^{j} \, \de_{\bar{\mathbf{n}}}\evbhsi \, h^{gi}_{ij} - \bn^i \bv^{j} \, \de_{\bar{\mathbf{v}}}\evbhsi \, h^{gi}_{ij} - \half \, \bn^i \bn^j \, \h{h}^{gi}_{ij} \, \big( \bn^k \bv_{k}^{\p} \big) - \\[2mm]
- \bn_i \, \boxfi \, \hdvf^{i}_{gi} \bigg) + M_5^3 \, \bigg( \frac{1}{2} \, \bv^{i} \bv^{j} \, \, \boxf \h{h}^{gi}_{ij} + \frac{3}{2} \, \hp^{\p \p} + \boxf \, \hp - \bv_{i}^{\p} \, \boxf \, \hdvf^{i}_{gi} \bigg) = \hat{\mcal{T}}_{tr} - f(\hxi) \, \dla \label{gloriagi}
\end{multline}
Taking the trace of the ``derivative'' equation (\ref{dermunujunctioncondgi}) we get
\beq
\label{dermunujunctioncondgibis}
2 M_6^4 \, \bn_i \, \boxf \, \hdvf^{i}_{gi} + M_5^3 \bigg( \! - \frac{1}{2} \, \bv^{i} \bv^{j} \, \boxf \, \h{h}^{gi}_{ij} - \boxf \, \hp + \bv_{i}^{\p} \, \boxf \, \hdvf^{i}_{gi} \bigg) = \boxf \, \hat{\mcal{T}}_{de}
\eeq

\subsubsection{The pure tension case}
\label{The pure tension case}

As a first check of the consistency of our analysis, we consider the case of a pure tension perturbation. In this case, the energy-momentum for the matter vanishes $\h{\mcalT}_{ab} = 0$ and we have
\beq
\d \hat{T}_{ab} = - \d_{a}^{\, \, \m} \,\, \d_{b}^{\, \, \n} \, f(\hxi) \,\, \dla \,\, \e_{\m\n}
\eeq
Inspired by the form of the background solutions, we consider an ansatz for the perturbation fields where the bulk metric perturbation $h_{AB}$ and the 4D components of the bending $\hdvf_{\m}$ vanish, and the $z$ and $y$ component of the bending depend only on $\hxi$ (note that in this case the bending modes coincide with their gauge invariant versions). This implies that
\begin{align}
\h{h}_{\xi\xi} &= 2 \, \bv_i \, \hdvf^{i \, \p} & \h{h}_{\xi\m} &= 0
\end{align}
so the requirement that the coordinate system $(\hxi,\chd)$ is Gaussian Normal inside the cod-1 brane is equivalent to the condition $\bv_i \, \hdvf^{i \, \p} = 0$. It is easy to see that all the equations of motion are identically satisfied apart from (\ref{gloriagi}) which reads
\beq
\label{Eszter}
2 M_6^4 \,\, \bn_i \, \hdvf^{i \, \p \p}  = f(\hxi) \, \dla
\eeq
To solve this equation, it is useful to introduce the orthogonal and parallel component of the perturbation of the parallel vector $\mathbf{v}_{(\xi)}$, namely
\begin{align}
\label{deltavdefinitiongi}
\d v_{\perp} &\equiv \bn_i \, \dvf_{gi}^{i \, \p} & \d v_{\shortparallel} &\equiv \bv_i \, \dvf_{gi}^{i \, \p}
\end{align}
and in particular we note that in this case in cod-1 GNC we have $\d \h{v}_{\shortparallel} = 0$. Using the identity $\bn_i \, \hdvf^{i \, \p \p} = (\bn_i \, \hdvf^{i \, \p})^{\p} - \bn_{i}^{\p} \, \hdvf^{i \, \p}$ and the relation (\ref{Bobbi}), we can express the equation (\ref{Eszter}) as  
\beq
2 M_6^4 \,\, \dhvn^{\p}  = f(\hxi) \, \dla
\eeq
which can be integrated to give
\beq
\label{republic}
\dhvn(\hxi) = \frac{\dla}{2 M_6^4} \int_{0}^{\hxi} f(\z) \, d\z
\eeq
which in particular implies that $\dhvn(\hxi)$ is constant for $\hxi \gtrless \pm l_2$ . Since (again using (\ref{Bobbi})) we have in general
\beq
\dvfn^\p = - \Sp \, \dvfp + \dvn
\eeq
we conclude that for $\abs{\hxi} \geq l_2$ we have
\beq
\label{Broncos}
\hdvfn (\hxi) = \frac{\dla}{4 M_6^4} \, \abs{\hxi} + \d\!\h{\vf}_0
\eeq
where $\d\!\h{\vf}_0$ is an integration constant. Note that the equations of motion does not fix $\d\!\h{\vf}_0$, which is then arbitrary; this is consistent with the fact that, since the bulk is exactly Minkowsky, a rigid translation of the brane is a symmetry of the system.

To understand the geometrical meaning of this configuration, we note that the (total) embedding function is of the form
\beq
\vf^A(\hxi) = \Big( \mZ(\hxi), \mY(\hxi), 0,0,0,0 \Big)
\eeq
where $\mZ = Z + \hdvf^{z}$ and $\mY = Y + \hdvf^{y}$. The solution defined by (\ref{Broncos}) corresponds to a configuration where the bulk is a ($Z_2$ symmetric) couple of slices of the 6D Minkowski spacetime, such that the total deficit angle is $\a = 4 \vartheta_+$ where
\beq
\tan \vartheta_+ = \frac{d \mZ (\mY)}{d \mY}\bigg\rvert_{+} = \frac{d \mZ (\hxi)}{d \hxi}\bigg\rvert_{+} \, \bigg( \frac{d \mY (\hxi)}{d \hxi} \bigg)^{\!\! -1}\bigg\rvert_{+}
\eeq
which at first order in perturbations reads
\beq
\tan \vartheta_+ = \tan S_+ + \frac{1}{\cos^2 S_+} \, \dvn\Big\rvert_{+}
\eeq
Applying the $\arctan$ to both sides of the former equation and expanding it around the value $\tan S_+$ in the right hand side, we get at first order
\beq
\vartheta_+ = S_+ + \dvn\Big\rvert_{+} = \frac{\bla + \dla}{4 M_6^4}
\eeq
where we used the background relation (\ref{Dom}) and the perturbative solution (\ref{republic}) to obtain the second equality. We conclude that a pure tension perturbation $\dla$ on the cod-2 brane produces a variation of the deficit angle
\beq
\d \a = \frac{\dla}{M_6^4}
\eeq
while the bulk metric remains the Minkowski metric. This is the same result we get from the exact solutions we obtained in section \ref{Pure tension solutions}, and therefore suggests that our perturbative analysis is consistent.

\section{The thin limit of the codimension-2 brane}
\label{ThinLimitCod2Brane}

We now turn to the analysis of the thin limit of the ribbon cod-2 brane, i.e.~to the limit $l_2 \rightarrow 0^+$.

\subsubsection*{The thin limit procedure}

In general, the thin limit description of a theory with localized sources is a description which provides a very good approximation to the true theory when we focus on length scales which are much bigger than the typical localization scales of the sources. However, to perform the thin limit in practice it is more convenient to adopt a different (but equivalent) point of view: we consider a fixed theory, and consider source configurations whose thickness becomes smaller and smaller, while keeping constant the ``total'' amount of the source.

In our case, we construct a sequence of source configurations $\h{T}^{_{[n]}}_{ab}$ whose localization length $l_2^{_{[n]}}$ tends to zero when $n \rightarrow + \infty$, but such that the cod-2 energy momentum tensor (defined in (\ref{cristina})) is independent from $n$. In the perturbative decomposition, this implies that we consider a sequence of localizing functions $f_{^{[n]}}$ which is a realization of the Dirac delta function, so that the background energy-momentum tensor converges to a thin pure tension configuration
\beq
\bar{T}^{[n]}_{ab}(\hxi, \chd) \, \xrightarrow{n \rightarrow + \infty} \, - \, \d_{a}^{\, \, \m} \,\, \d_{b}^{\, \, \n} \,\, \d(\hxi) \, \bla \,\, \e_{\m\n}
\eeq
Likewise, we consider a sequence of matter energy-momentum tensors $\h{\mcalT}^{_{[n]}}_{ab}$ such that the matter cod-2 energy momentum tensor (defined in (\ref{surprise})) is independent from $n$; we still indicate it with $\mcalT^{_{(4)}}_{\m\n}$. For each value of $n$, we consider the embedding and bulk metric configurations which solve the equation of motion with $\h{T}^{_{[n]}}_{ab}$ as a source. In particular, to the sequence $\bar{T}^{_{[n]}}_{ab}$ we associate the sequence of background embeddings $Z_{^{[n]}}$ and $Y_{^{[n]}}$, and to the sequence $\h{\mcalT}^{_{[n]}}_{ab}$ we associate the sequence of bending modes $\dvf^{i}_{^{[n]}}$ and the sequence of metric perturbations $\mscr{H}^{_{[n]}}_{\m\n}$, $\pi_{^{[n]}}$ and so on. For economy of notation, we indicate with $h^{_{[n]}}_{ij}$ and $\dvf^{i}_{^{[n]}}$ respectively the sequences associated to the gauge invariant variables $h^{_{gi}}_{ij}$ and $\dvf^{i}_{^{gi}}$ and not to $h_{ij}$ and $\dvf^{i}$ themselves. To indicate the limit configurations to which to sequences of fields converge, we substitute the symbol $[n]$ with the symbol $\infty$, so for example $Z_{^{[n]}} \rightarrow Z_{\infty}$, $\pi_{^{[n]}} \rightarrow \pi_{\infty}$, $\dvf^{i}_{^{[n]}} \to \dvf^{i}_{\infty}$ and so on. The thin limit of the background configurations is studied in appendix \ref{Thin limit of the background}.

Note that, at linear order in perturbations, the effect on the metric perturbations and on the bending modes of a tension perturbation and of a matter perturbation is additive, so we may write
\beq
\dvf^i = \dvf^i_{pt} + \dvf^i_{pm}
\eeq
where $\dvf^i_{^{pt}}$ is the bending correspondent to a pure tension perturbation and $\dvf^i_{^{pm}}$ is the bending correspondent to a pure matter perturbation (the perturbation of the metric vanishes in the pure tension case, so $h_{_{AB}} = h_{^{AB}}^{_{pm}}$). We already studied the effect of a pure tension source perturbation in our framework in section \ref{The pure tension case}: therefore, from now on we will consider only a pure matter perturbation, and will implicitly assume this even if, for economy of notation, we will omit to write explicitly the subscript/superscript ``$_{pm}/^{pm}\,$''.

\subsection{The singular structure of the perturbations f\mbox{}ields}

To understand what happens to the solutions of the equations of motion when the cod-2 brane becomes thin, the only thing we can do is to make hypothesis on the behaviour of the perturbation fields when the thin limit is taken, and a posteriori study if these assumptions are consistent. Since we want to investigate if gravity remains finite in the thin limit, we propose here an ansatz which generates a finite gravitational field on the thin cod-2 brane.

\subsubsection{The geometric ansatz}
\label{The geometric ansatz}

Naively speaking, the non-trivial point in the phenomenon of gravity regularization in our set-up is that, when we perform the thin limit, we need to generate a delta-function singularity in the left-hand side of the equations (\ref{dermunujunctioncondgi}) and (\ref{gloriagi}) (therefore, in the extrinsic curvature and/or in the cod-1 Einstein tensor), while having a finite cod-2 Ricci tensor. We mentioned above that there exists a class of configurations which has exactly these properties, i.e.~the configurations where the bulk metric is smooth while the embedding converges to a profile which is cuspy in the $\hxi$-direction at $\mcal{C}_2$. These configurations satisfy the requirements because the delta singularity is produced via second $\hxi$-derivatives of the embedding function, which being continuous gives rise to an induced metric which is finite on the cod-1 brane even at the position of the cod-2 brane.

Under this hypothesis, the only terms which can diverge in the cod-1 extrinsic curvature and Einstein tensor are those which contain second derivatives with respect to $\hxi$ of the embedding functions (background or perturbation part), and the second derivatives with respect to $\hxi$ of the bulk perturbations evaluated on the brane (e.g.~$\hp^{\p\p}$ and $\h{\mscr{H}}_{\m\n}^{\p\p}$), while the evaluation on the brane of the bulk perturbations and their derivatives with respect to the bulk coordinates do not diverge (for example $\dezq \pi$ and $\dezq \pi$ evaluated in $\vfd(\xid)$ are bounded). Taking a look at the sections \ref{Ricci and Einstein tensors} and \ref{Extrinsic curvature} of the appendix \ref{IndMetrandCurvTensApp}, it is clear that the only components of the cod-1 curvature tensors $\d\h{G}_{ab}$ and $\d\!\h{K}_{ab}$ that can diverge are $\d\h{G}_{\m\n}$ and $\d\!\h{K}_{\xi\xi}$. This is promising, since $\d\h{G}_{\m\n}$ and $\d\!\h{K}_{\xi\xi}$ appear only in the $\m\n$ components of the junction conditions (\ref{formula interessante3}) which are the only components which are sourced by a diverging energy-momentum tensor in the thin limit. Furthermore, taking a look at the junction condition for the tensor sector (\ref{maguardaunpo}) and at the ``derivative'' and ``trace'' junction conditions for the scalar sector (\ref{dermunujunctioncondgi}) and (\ref{gloriagi}), it is evident that each of these equations possesses on the left hand side a diverging term which may balance the diverging coming from the energy-momentum on the right hand side. It is worthwhile to notice that, both for the tensor and the ``derivative'' equation, the diverging pieces are contributed solely by the cod-1 induced gravity term: in absence of induced gravity on the cod-1 brane, this ansatz would not be valid, apart from the special case of a pure tension perturbation. We further comment on this point in the Conclusions.

We therefore choose this ansatz as a working hypothesis on the structure of the perturbation fields in the thin limit. Since the pillbox integration across the cod-2 brane is subtle, it turns out to be necessary to spell out very clearly the character of convergence of the sequences of perturbation fields to their limit configuration. Consistently with the considerations above, we assume that the perturbation of the bulk metric and its derivatives of every order converge \emph{uniformly} to smooth limit functions, and that the perturbation of the components of the embedding and all its 4D derivatives converge \emph{uniformly} to continuous limit functions. On the other hand, we assume that the first derivative with respect to $\hxi$ of the components of the embedding converge \emph{pointwise} to limit functions which are not necessarily continuous in $\hxi = 0$. This is consistent with the convergence properties of the sequences of background embedding functions $Z_{^{[n]}}$ and $Y_{^{[n]}}$, which we discuss in appendix \ref{Thin limit of the background} (for the definition of the standard concepts of uniform convergence and pointwise convergence of a sequence of functions see e.g.~\cite{Rudin}).

\subsubsection{Pure cod-1 and cod-2 junctions conditions}

From the point of view of the structure of the equations of motion, the thin limit has the effect of splitting the junction condition into two sets of conditions. In fact, when the cod-2 brane is thick we have to solve the equations both outside the physical cod-2 brane, where $\bv^i$ and $\bn_i$ are constant (external solutions), and inside the cod-2 brane (internal solution), and join smoothly these solutions at the boundaries of the cod-2 brane. Since $l_2^{_{[n]}} \rightarrow 0$ for $n \rightarrow + \infty$, in the thin limit the equations for the ``external'' fields are valid for $\hxi \neq 0$: the equations for the internal fields then generate a set of conditions which relate the value of the external fields at $\hxi = 0^-$ and $\hxi = 0^+$. We refer to the former set of equations as \emph{pure codimension-1 junction conditions} and to the latter set of conditions as \emph{codimension-2 junction conditions}.

To derive the codimension-2 junction conditions, it is necessary to make explicit use of the parity properties of the fields with respect to the reflection $\hxi \rightarrow - \hxi$. Note first of all that, by construction, $\Zp$ is odd while $\Yp$ is even. As we mentioned in section \ref{The equations of motion}, the $\mathbb{Z}_2$ symmetry present inside the cod-1 brane implies that the $\xi\xi$ and $\m\n$ components of $\h{\mathbf{g}}$, $\h{\mathbf{R}}$ and $\h{\mathbf{K}}$ are even, while their $\xi\m$ components are odd. From (\ref{color1})--(\ref{color6}) it follows that $\h{h}^{_{gi}}_{zz}$, $\h{h}^{_{gi}}_{yy}$, $\hp$, $\hdvf^{z}_{^{gi}}$, $\h{\mcalA}_{z\m}$ and $\h{\mscr{H}}_{\m\n}$ are even, while $\h{h}^{_{gi}}_{zy}$, $\hdvf^{y}_{^{gi}}$ and $\h{\mcalA}_{y\m}$ are odd. Furthermore, every $\dehxi$ derivative changes the parity from even to odd and the other way around, while the 4D derivatives $\de_{\chm}$ leave the parity unaltered. Expressing the $\dehxi$ derivative as $\bv^i \de_i$, we can then infer that $\dez$ derivatives acting on a bulk field leave unaltered the parity, while $\dey$ changes it. We stress that this is true only when the bulk fields and their $\de_i$ derivatives are evaluated \emph{on the cod-1 brane}, since we don't impose a $\mathbb{Z}_2 \times \mathbb{Z}_2$ symmetry in the bulk and therefore the bulk fields do not have definite parity properties away from the cod-1 brane.

\subsection{Tensor and vector sectors}

\subsubsection{Tensor sector}

Let's consider first the tensor sector. Considering the equation (\ref{maguardaunpo}), we first note that, for $\abs{\hxi} \geq l_2^{_{[n]}}$, each element of the sequence $\mscr{H}^{_{[n]}}_{\m\n}$ obeys

\beq
2 \Msf \, \de_{\bar{\mathbf{n}}_{[n]}}\evbhsin \, \mscr{H}^{[n]}_{\m\n} + \Mft \,\, \boxfi \, \h{\mscr{H}}^{[n]}_{\m\n} = 0
\eeq
Taking the limit $n \rightarrow + \infty$, the limit of the sequence (if it exists) satisfies for $\hxi \neq 0$ the pure cod-1 junction condition

\beq
2 \Msf \, \de_{\bar{\mathbf{n}}_{_{\infty}}}\evbhsinf \, \mscr{H}^{\infty}_{\m\n} + \Mft \,\, \boxfi \, \h{\mscr{H}}^{\infty}_{\m\n} = 0
\eeq

To derive the cod-2 junction condition, we substitute $\bar{\mathbf{n}}$, $\mscr{H}_{\m\n}$ and $\mscrT_{\m\n}$ respectively with $\bar{\mathbf{n}}^{_{[n]}}$, $\mscr{H}^{_{[n]}}_{\m\n}$ and $\mscrT^{_{[n]}}_{\m\n}$ in the equation (\ref{maguardaunpo}) and perform a pillbox integration across the cod-2 brane. Since by our ansatz both $\bn_{^{[n]}}^i$ and $\de_i \, \mscr{H}^{_{[n]}}_{\m\n}$ remain bounded in the $n \rightarrow + \infty$ limit, we get
\beq
\lim_{n \rightarrow + \infty} \int_{-l_2^{[n]}}^{+l_2^{[n]}} \! d\hxi \,\, \bigg[ \, 2 \Msf \, \bn^{i}_{[n]} \, \de_{i}\evbhsin \, \mscr{H}_{\m\n}^{[n]} + \Mft \,\, \boxfi \, \h{\mscr{H}}_{\m\n}^{[n]} \, \bigg] = \Mft \, \bigg[ \dehxi \, \h{\mscr{H}}_{\m\n}^{\infty} \, \bigg]_{0^{\pm}}
\eeq
and therefore we obtain the cod-2 junction condition for the tensor sector
\beq
\Mft \, \de_{\hxi}\Big\rvert_{0^+} \h{\mscr{H}}^{\infty}_{\m\n} = - \mscrT^{(4)}_{\m\n}
\eeq

\subsubsection{Vector sector - gauge invariant variables formulation}

Regarding the vector sector, for $\abs{\hxi} \geq l_2^{_{[n]}}$, each element of the sequences $\h{\mcalA}^{_{[n]}}_{z\m}$ and $\h{\mcalA}^{_{[n]}}_{y\m}$ obeys

\begin{align}
2 M_6^4 \, \Big( \bv^i_{[n]} \, \de_{\bar{\mathbf{n}}_{[n]}}\evbhsin \, \mcalA_{i\m}^{[n]} - \bn^i_{[n]} \, \h{\mcalA}_{i\m}^{[n] \, \p} \Big) + M_5^3 \, \boxf \, \bv^i_{[n]} \, \h{\mcalA}_{i\m}^{[n]} &= 0 \label{Vectorgi1} \\[3mm]
2 M_6^4 \, \bn^i_{[n]} \, \h{\mcalA}^{[n]}_{i\m} + M_5^3 \, \bv^i_{[n]} \, \h{\mcalA}^{[n] \, \p}_{i\m} &= 0 \label{Vectorgi2}
\end{align}
where the eq.~(\ref{Vectorgi1}) comes from the $\xi\m$ components of the junction conditions eq.~(\ref{avventura1}) and the eq.~(\ref{Vectorgi2}) comes from the $\m\n$ components of the junction conditions eq.~(\ref{Sara}). Therefore, the pure cod-1 junction conditions for the vector sector are
\begin{align}
2 M_6^4 \, \bigg( \bv^i_{\infty} \, \de_{\bar{\mathbf{n}}_{_{\infty}}}\evbhsinf \, \mcalA_{i\m}^{\infty} - \bn^i_{\infty} \, \h{\mcalA}_{i\m}^{\infty \, \p} \bigg) + M_5^3 \, \boxf \, \bv^i_{\infty} \, \h{\mcalA}_{i\m}^{\infty} &= 0 \label{VectorGI1} \\[3mm]
2 M_6^4 \, \bn^i_{\infty} \, \h{\mcalA}^{\infty}_{i\m} + M_5^3 \, \bv^i_{\infty} \, \h{\mcalA}^{\infty \, \p}_{i\m} &= 0 \label{VectorGI2}
\end{align}
which are again valid for $\hxi \neq 0$.

To derive the cod-2 junction conditions, we note that the $\m\n$ equation (\ref{Sara}) contains diverging pieces while the $\xi\m$ equation (\ref{avventura1}) contains at most discontinuous pieces, and its source term $\h{D}_{\m}$ vanishes in the thin limit. Therefore, the cod-2 junction conditions are obtained by imposing that the left hand side of the $\xi\m$ pure cod-1 junction condition (\ref{VectorGI1}) is continuous at $\hxi = 0$ (i.e.~its $\hxi \to 0^+$ and $\hxi \to 0^-$ limit coincide), and by performing a pillbox integration of the $\m\n$ equation (\ref{Sara}). The former condition, using $\de_{\hxi} = \bv^i \de_i$, can be rewritten as
\beq
\label{VectorGG}
\bigg[ 2 M_6^4 \, \Big( Z^{\p \, 2}_{\infty} + Y^{\p \, 2}_{\infty} \Big) \bigg( \de_{z}\evbhsinf \, \mcalA_{y\m}^{\infty} - \de_{y}\evbhsinf \, \mcalA_{z\m}^{\infty} \bigg) + M_5^3 \, \boxf \, \Big( \Zp_{\infty} \, \h{\mcalA}_{z\m}^{\infty} + \Yp_{\infty} \, \h{\mcalA}_{y\m}^{\infty} \Big) \bigg]_{0^{\pm}} = 0
\eeq
where $[\phantom{a}]_{0^{\pm}}$ means $\vert_{0^+} - \vert_{0^-}$ as usual. Note that the left hand side of this equation is a linear combination of terms which are odd with respect to the reflection symmetry $\hxi \rightarrow -\hxi$ and continuous\footnote{The functions $Z^{\p \, 2}_{\infty}$ and $Y^{\p \, 2}_{\infty}$ are not continuous at $\hxi = 0$, but their product with an odd and continuous function is.}, with the only exception of the term $M_5^3 \, \boxf \, \Zp_{\infty} \, \h{\mcalA}_{z\m}^{\infty}$ which is odd and discontinuous at $\hxi = 0$. Since every odd and continuous term vanishes both in $\hxi = 0^+$ and in $\hxi = 0^-$, the condition (\ref{VectorGG}) gives
\beq
M_5^3 \, \boxf \, \Zp_{\infty}\Big\rvert_{0^+} \, \h{\mcalA}_{z\m}^{\infty}\Big\rvert_{\hxi = 0} = 0
\eeq
which implies
\beq
\label{Ave}
\h{\mcalA}_{z\m}^{\infty}\Big\rvert_{\hxi = 0} = 0
\eeq
On the other hand, taking into account (\ref{nottoobad2}), the pillbox integration of the equation (\ref{Sara}) reads
\beq
\lim_{n \to \infty} \int_{-l_2^{[n]}}^{l_2^{[n]}} \bigg[ 2 M_6^4 \,\, \bn^{i}_{[n]} \, \h{\mcalA}^{[n]}_{i\m} + M_5^3 \, \de_{\hxi} \, \Big( \bv^{i}_{[n]} \, \h{\mcalA}^{[n]}_{i\m} \Big) \bigg] d\hxi  = 0
\eeq
and, since by our ansatz both $\bn_{^{[n]}}^i$ and $\h{\mcalA}^{[n]}_{i\m}$ remain bounded in the $n \rightarrow + \infty$ limit, we get
\beq
M_5^3 \, \bigg[ \bv^{i}_{\infty} \, \h{\mcalA}^{\infty}_{i\m} \bigg]_{0^{\pm}} = 0
\eeq
Since $\h{\mcalA}^{\infty}_{y\m}$ vanishes in $\hxi = 0^{\pm}$ because it is odd and continuous, we get
\beq
2 M_5^3 \, \Zp_{\infty}\Big\rvert_{0^+} \, \h{\mcalA}^{\infty}_{z\m}\Big\rvert_{\hxi = 0} = 0
\eeq
and so we reproduce (\ref{Ave}). Therefore, the gauge invariant vector modes vanish on the cod-2 brane in the thin limit: this is compatible with the observation at the end of section \ref{source decomposition} that the vector modes are not sourced in the thin limit.

\subsubsection{Vector sector - master variable formulation}
\label{Vector sector - master variable formulation}

Using the relations (\ref{vectormasterdef}) and (\ref{formula pero2}), we can recast the thin limit equations of motion for the vector sector in terms of the master variable $\Phi_{\m}^{\infty}$. Considering first the pure cod-1 junction conditions, the equation (\ref{VectorGI1}) in terms of $\Phi_{\m}^{\infty}$ reads
\beq
\label{VectorMaster}
2 M_6^4 \, \bigg( \de^{2}_{\bar{\mathbf{n}}_{_{\infty}}}\evbhsinf \, \Phi_{\m}^{\infty} + \de^{2}_{\bar{\mathbf{v}}_{_{\infty}}}\evbhsinf \, \Phi_{\m}^{\infty} \bigg) + M_5^3 \, \boxf \, \de_{\bar{\mathbf{n}}_{_{\infty}}}\evbhsinf \, \Phi_{\m}^{\infty} = 0
\eeq
Since $\bar{\mathbf{v}}_{_{\infty}}$ and $\bar{\mathbf{n}}_{_{\infty}}$ are orthonormal, we have
\beq
\de^{2}_{\bar{\mathbf{n}}_{_{\infty}}} + \de^{2}_{\bar{\mathbf{v}}_{_{\infty}}} = \de^{2}_{z} + \de^{2}_{y} \equiv \bigtriangleup_{2}
\eeq
and therefore we can express (\ref{VectorMaster}) as
\beq
\label{vectorpurecod1a}
2 M_6^4 \, \bigtriangleup_{2}\evbhsinf \, \Phi_{\m}^{\infty} + M_5^3 \, \boxf \, \de_{\bar{\mathbf{n}}_{_{\infty}}}\evbhsinf \, \Phi_{\m}^{\infty} = 0
\eeq
where $\bigtriangleup_{2}$ is the 2D Laplace operator. Similarly, the equation (\ref{VectorGI2}) in terms of $\Phi_{\m}^{\infty}$ reads
\beq
- 2 M_6^4 \, \de_{\bar{\mathbf{v}}_{_{\infty}}}\evbhsinf \, \Phi_{\m}^{\infty} + M_5^3 \, \de_{\bar{\mathbf{n}}_{_{\infty}}}\evbhsinf \, \de_{\bar{\mathbf{v}}_{_{\infty}}}\evbhsinf \, \Phi_{\m}^{\infty} = 0
\eeq
or equivalently
\beq
\label{vectorpurecod1b}
\de_{\hxi} \, \bigg( 2 M_6^4 \, \h{\Phi}_{\m}^{\infty} - M_5^3 \, \de_{\bar{\mathbf{n}}_{_{\infty}}}\evbhsinf \, \Phi_{\m}^{\infty} \bigg) = 0
\eeq

Regarding the cod-2 junction conditions, the conditions
\begin{align}
\h{\mcalA}_{z\m}^{\infty}\Big\rvert_{\hxi = 0} &= 0 & \h{\mcalA}_{y\m}^{\infty}\Big\rvert_{\hxi = 0} &= 0
\end{align}
imply that the $z$ and $y$ derivatives of $\Phi_{\m}^{\infty}$ vanish on the cod-2 brane
\begin{align}
\de_{z}\Big\rvert_{\bvfd(0,\chd)} \, \Phi_{\m}^{\infty} &= 0 & \de_{y}\Big\rvert_{\bvfd(0,\chd)} \, \Phi_{\m}^{\infty} &= 0
\end{align}

\subsection{Scalar sector}

\subsubsection{Pure cod-1 junctions conditions}

To derive the pure cod-1 junction conditions for the scalar sector, it is sufficient to substitute in the equations (\ref{xixijunctioncondgi}), (\ref{ximujunctioncondgibis}), (\ref{gloriagi}) and (\ref{dermunujunctioncondgibis}) the background and perturbation fields with the correspondent sequences indexed by $[n]$, to impose $\bv_i^{\p} = 0$ and to set to zero the right hand side of these equations. Taking the limit $n \to \infty$ we then obtain the pure cod-1 junction conditions

\noi ($\xi\xi$ component)
\beq
\label{Gingerxixi}
2 \, M_6^4 \, \bigg( 2 \, \de_{\bar{\mathbf{n}}_{\infty}}\evbhsinf \, \pi_{\infty} - \boxf \, \hdvfnin \bigg) + \frac{3}{2} \, M_5^3 \, \boxf \, \hp_{\infty} = 0
\eeq

\noi ($\xi\m$ components)
\beq
\label{Gingerximu}
2 \, M_6^4 \, \bigg( \half \, \bn^i_{\infty} \bv^{j}_{\infty} \, \h{h}^{\infty}_{ij} + \hdvfninp \, \bigg) - \frac{3}{2} \, M_5^3 \, \, \hp_{\infty}^{\p} = 0
\eeq

\noi ($\m\n$ components, derivative)
\beq
\label{Gingerdermunu}
2 M_6^4 \, \boxf \, \hdvfnin + M_5^3 \bigg( \! - \frac{1}{2} \, \bv^{i}_{\infty} \bv^{j}_{\infty} \, \boxf \, \h{h}^{\infty}_{ij} - \boxf \, \hp_{\infty} \bigg) = 0
\eeq

\noi ($\m\n$ components, trace)
\begin{multline}
\label{Gingertracemunu}
2 M_6^4 \, \bigg( \frac{3}{2} \, \de_{\bar{\mathbf{n}}_{\infty}}\evbhsinf \, \pi_{\infty} + \frac{1}{2} \, \bv^{i}_{\infty} \bv^{j}_{\infty} \, \de_{\bar{\mathbf{n}}_{\infty}}\evbhsinf \, h^{\infty}_{ij} - \bn^i_{\infty} \bv^{j}_{\infty} \, \de_{\bar{\mathbf{v}}_{\infty}}\evbhsinf \, h^{\infty}_{ij} - \boxfi \, \hdvfnin \bigg) + \\[1mm]
+ M_5^3 \, \bigg( \frac{1}{2} \, \bv^{i}_{\infty} \bv^{j}_{\infty} \, \, \boxf \h{h}^{\infty}_{ij} + \frac{3}{2} \, \hp^{\p \p}_{\infty} + \boxf \, \hp_{\infty} \bigg) = 0
\end{multline}
which are valid for $\hxi \neq 0$.

These four equations are actually not independent, but are linked by differential relations if we take into account the bulk equations. In fact, expressing the equations in terms of the metric master variable $\pi$ (using the relation (\ref{desiree2}), which encodes part of the bulk equations), it is possible to see that the equations above are linked by the relations
\begin{align}
\de_\xi \, (\ref{Gingerxixi}) + \boxf (\ref{Gingerximu}) &= 0 \\[2mm]
\de_\xi \, (\ref{Gingerximu}) + (\ref{Gingerdermunu}) + (\ref{Gingertracemunu}) &= 0
\end{align}
This implies that only two of the four equations (\ref{Gingerxixi})--(\ref{Gingertracemunu}) are independent: we choose to work with the equation $(\ref{Gingerxixi}) + (\ref{Gingerdermunu})$ and with the (5D trace) equation $(\ref{Gingerxixi}) + (\ref{Gingerdermunu}) + 4 \,\, (\ref{Gingertracemunu})$. Using the relations  
\begin{align}
\bv^{i} \bv^{j} \, \de_{i} \, \de_{j} \, \pi \evbhsi &= \hp^{\p \p} \\[1mm]
\bv^{i} \bv^{j} \, \bn^{k} \,\, \de_{i} \, \de_{j} \, \de_{k} \, \pi \evbhsi &= \de^2_{\hxi} \, \bigg( \de_{\bar{\textbf{n}}} \, \pi\evbhsi \bigg)
\end{align}
which are valid where $\bv_i^{\p} = \bn_i^{\p} = 0$, these two equations read in terms of the master variables

\beq
\label{purecod1pi}
2 \, M_6^4 \, \de_{\bar{\textbf{n}}_{\infty}}\evbhsinf \, \pi_{\infty} + M_5^3 \, \boxfi \, \hp_{\infty} = 0
\eeq
and

\beq
\label{purecod1bending}
\boxfi \, \hdvfnin = \half \, \de_{\bar{\textbf{n}}_{\infty}}\evbhsinf \, \pi_{\infty} + 2 \, \de^2_{\hxi} \, \bigg( \frac{\de_{\bar{\textbf{n}}_{\infty}}}{\boxf}\evbhsinf \, \pi_{\infty} \bigg)
\eeq

\noi Note that the field $\pi_{\infty}$ obeys a decoupled equation also on the pure cod-1 brane.

\subsubsection{Cod-2 junctions conditions}

To derive the cod-2 junction conditions, we note that the left hand sides of the $\xi\xi$ equation (\ref{xixijunctioncondgi}) and of the $\xi\m$ equation (\ref{ximujunctioncondgibis}) do not contain terms which diverge in the thin limit, according to our ansatz, and their source terms vanish in the thin limit. On the other hand, the left hand side of the $\m\n$ ``derivative'' equation (\ref{dermunujunctioncondgibis}) and of the $\m\n$ ``trace'' equation (\ref{gloriagi}) contain diverging terms, and the pillbox integration of their source terms remains non-vanishing in the thin limit. Therefore we impose that the left hand side of the $\xi\xi$ and $\xi\m$ pure junction continuous (\ref{Gingerxixi}) and (\ref{Gingerximu}) are continuous in $\hxi = 0$, while we perform a pillbox integration of the equations (\ref{gloriagi}) and (\ref{dermunujunctioncondgibis}).

Regarding the equation (\ref{Gingerxixi}), the condition 
\beq
\bigg[ 2 \, M_6^4 \, \bigg( 2 \, \bn^i_{\infty} \, \de_{i}\evbhsinf \, \pi_{\infty} - \boxf \, \hdvfnin \bigg) + \frac{3}{2} \, M_5^3 \, \boxf \, \hp_{\infty} \bigg]_{0^{\pm}} = 0
\eeq
is identically satisfied since all the terms in the left hand side are even. Concerning the equation (\ref{Gingerximu}), the terms in the left hand side are odd and so the condition
\beq
\bigg[ 2 \, M_6^4 \, \bigg( \half \, \bn^i_{\infty} \bv^{j}_{\infty} \, \h{h}^{\infty}_{ij} + \hdvfninp \, \bigg) - \frac{3}{2} \, M_5^3 \, \, \hp_{\infty}^{\p} \bigg]_{0^{\pm}} = 0
\eeq
is not identically satisfied, and produces the cod-2 junction condition
\beq
\label{consistency}
\Bigg[ M_6^4 \, \sin \bigg( \frac{\bla}{2 \Msf} \bigg) \Big( \h{h}^{\infty}_{zz} - \h{h}^{\infty}_{yy} \Big) + 4 \, M_6^4 \, \hdvfn^{\infty \, \p} - 3 \, M_5^3 \, \hp_{\infty}^{\p} \Bigg]_{0^+} = 0
\eeq
where we used the (background) relation $2 \, \Zp_{\infty}\rvert_{0^+} \, \Yp_{\infty}\rvert_{0^+} = \sin \big( \bla/2 \Msf \big)$ and the fact that $\h{h}^{\infty}_{zy}$ is continuous and odd and therefore vanishes in $\hxi = 0^+$.

We come now to the pillbox integration of the equations (\ref{gloriagi}) and (\ref{dermunujunctioncondgibis}) across the cod-2 brane. According to our ansatz on the singular behaviour of the perturbations at the cod-2 brane, the only term which diverges in the left hand side of the equation (\ref{dermunujunctioncondgibis}) is $\Mft \, \bv_{i}^{\p} \, \boxf \, \hdvf^{i}_{^{gi}}$. Therefore, the derivative part of the $\m\n$ components of the junction conditions produces the condition
\beq
\label{Aloha der}
M_5^3 \lim_{n \rightarrow + \infty} \int_{-l_2^{[n]}}^{+l_2^{[n]}} \!\! d \hxi \,\,\, \bv_{i}^{[n] \, \p} \, \boxf \, \hdvf^{i}_{[n]} = \boxf \, \mcal{T}^{(4)}_{de} 
\eeq
Similarly, the only terms which diverge in the left hand side of the equation (\ref{gloriagi}) are the ones which are derived twice with respect to $\hxi$ (remember that $\bv^{i \, \p} = \bvf^{i \, \p \p}$). Therefore, the trace part of the $\m\n$ components of the junction conditions produces the condition
\begin{multline}
\label{Aloha trace}
- M_6^4 \lim_{n \rightarrow + \infty} \int_{-l_2^{[n]}}^{+l_2^{[n]}} \! d \hxi \, \bigg( \bn_{[n]}^i \bn_{[n]}^j \, \Big( \bn^{[n]}_k \bv_{[n]}^{k \, \p} \Big) \, \h{h}^{[n]}_{ij} + 2 \, \bn^{[n]}_i \, \hdvf^{i \, \p\p}_{[n]} \bigg) + \\[2mm]
+ M_5^3 \lim_{n \rightarrow + \infty} \int_{-l_2^{[n]}}^{+l_2^{[n]}} \! d \hxi \, \bigg( \frac{3}{2} \, \hp_{[n]}^{\p \p} - \bv_{i}^{[n] \, \p} \, \boxf \, \hdvf^{i}_{[n]} \bigg) = \mcal{T}^{(4)}_{tr} 
\end{multline}
We can simplify this equation noting that, using (\ref{daje}), we have $\bn^{_{[n]}}_k \bv_{^{[n]}}^{k \, \p} = S_{^{[n]}}^{\p}$, and that the integral on the left hand side of (\ref{Aloha der}) is present also in (\ref{Aloha trace}). Using the continuity equation for the cod-2 energy-momentum tensor (\ref{nottoobad1}), we conclude that it is equivalent to impose the conditions (\ref{Aloha der}) and (\ref{Aloha trace}), or the condition (\ref{Aloha der}) and the following condition
\beq
\label{Aloha trace2}
M_6^4 \lim_{n \rightarrow + \infty} \int_{-l_2^{[n]}}^{+l_2^{[n]}} \! d \hxi \, \bigg( S_{[n]}^{\p} \, \bn_{[n]}^i \bn_{[n]}^j \, \h{h}^{[n]}_{ij} + 2 \, \bn^{[n]}_i \, \hdvf^{i \, \p \p}_{[n]} \bigg) - 3 \, M_5^3 \, \hp_{\infty}^{\p}\Big\rvert_{0^+} = 0 
\eeq
The integrations in the equations (\ref{Aloha der}) and (\ref{Aloha trace2}) are performed explicitly in the appendix \ref{Pillbox integration across the cod2 brane}. It turns out that the equation (\ref{Aloha trace2}) reproduces exactly the condition (\ref{consistency}), which is a confirmation of the consistency of our analysis, while the equation (\ref{Aloha der}) produces the condition
\beq
2 \, \Mft \, \tan \bigg( \frac{\bla}{4 \Msf} \bigg) \, \boxf \, \hdvfn^{\infty}\Big\rvert_{0^+} = \boxf \, \mcal{T}^{(4)}_{de} 
\eeq
The latter equation, together with (\ref{consistency}), constitutes the thin limit cod-2 junction conditions. Concerning the discussion at the end of section \ref{Scalar master variables} on the number of bending master variables, we note that we managed to express the cod-2 junction conditions in terms of $\dvfn$ only (and not $\dvfp$), and so we confirm that the thin limit equations for the scalar sector can be expressed in terms of one metric master variable, $\pi_{\infty}$, and one bending master variable, $\dvfn^{\infty}$. The equation (\ref{consistency}) in terms of the master variables reads
\beq
\Bigg[ 4 \, M_6^4 \, \hdvfn^{\infty \, \p} - 4 \, M_6^4 \, \sin \bigg( \frac{\bla}{2 \Msf} \bigg) \frac{\big( \de_z^2 - \de_y^2 \big)}{\boxf} \, \pi_{\infty}\evbbhsinf - 3 \, M_5^3 \, \hp_{\infty}^{\p} \Bigg]_{0^+} = 0 
\eeq

\subsection{Discussion}

We now discuss the thin limit equations we derived in the previous sections, in relation with the regularization of gravity and the well-definiteness of the thin limit. For clarity of exposition, in this section we omit the symbol $\infty$: it is implicitly assumed that all the fields which appear here are the limit configurations of the sequences indexed by $n$.

\subsubsection{Scalar sector}

Let's start with the discussion of the scalar sector. As we showed in the last section, the thin limit equations of motion for this sector read

\beq
\label{scalarthinbulk}
\boxs \, \pi = 0 \qquad (\textrm{bulk})
\eeq
\vspace{3mm}
\beq
\left.
\begin{aligned}
\label{scalarthinpurecod1}
2 \, M_6^4 \, \de_{\bar{\textbf{n}}}\evbhsi \, \pi + M_5^3 \, \boxfi \, \hp = 0\,\,& \\[3mm]
\boxfi \, \hdvfn = \half \, \de_{\bar{\textbf{n}}}\evbhsi \, \pi + 2 \, \de^2_{\hxi} \, \bigg( \frac{\de_{\bar{\textbf{n}}}}{\boxf}\evbhsi \, \pi \bigg)&
\end{aligned}
\quad \right\} \qquad (\textrm{pure cod-1 brane})
\eeq
\vspace{6mm}
\beq
\left.
\begin{aligned}
\label{scalarthincod2}
\Bigg[ 4 \, M_6^4 \, \hdvfn^{\p} - 4 \, M_6^4 \, \sin \bigg( \frac{\bla}{2 \Msf} \bigg) \frac{\big( \de_z^2 - \de_y^2 \big)}{\boxf} \, \pi\evbbhsi - 3 \, M_5^3 \, \hp^{\p} \Bigg]_{0^+} \! = 0\,\,& \\[3mm]
2 \, \Mft \, \tan \bigg( \frac{\bla}{4 \Msf} \bigg) \, \boxf \, \hdvfn\Big\rvert_{0^+} = \boxf \, \mcal{T}^{(4)}_{de}&
\end{aligned}
\quad \right\} \qquad (\textrm{cod-2 brane})
\eeq
Considering for the moment only the bulk equation (\ref{scalarthinbulk}) and the pure cod-1 junction conditions (\ref{scalarthinpurecod1}), we note that these equations are \emph{exactly} the equations for a 6D cod-1 DGP model: their form is more complicated than the one usually found in the literature just because we didn't impose transverse-traceless conditions in the bulk, which is usually a standard choice. In particular, the second of the equations (\ref{scalarthinpurecod1}) is the 5D wave equations for the bending mode $\hdvfn$, while the equation (\ref{scalarthinbulk}) and the first of the equations (\ref{scalarthinpurecod1}) are the 6D wave equation for $\pi$ together with the boundary condition on the brane. With appropriate boundary conditions at infinity in the extra dimensions, if we assume that the fields $\hdvfn$, $\pi$ and $\hp$ are smooth then the system of these equations is well-defined, and admits a unique solution. If now we relax the assumption that $\hdvfn$, $\pi$ and $\hp$ are smooth, and assume that in $\hxi = 0$ the fields $\hdvfn$ and $\hp$ are continuous but not derivable with respect to $\hxi$, the system of differential equations (\ref{scalarthinbulk})--(\ref{scalarthinpurecod1}) does not single out a unique solution any more, because there is freedom in choosing how to patch the solutions for the fields $\hdvfn$ and $\hp$ at $\hxi = 0^+$ and $\hxi = 0^-$. To single out a unique solution, we need to add two conditions (one for $\hp$ and one for $\hdvfn$) at $\hxi = 0$ which fix this arbitrariness.

This is exactly the role of the cod-2 junction conditions in our set-up: in fact, it is convenient to group the equations (\ref{scalarthinbulk})--(\ref{scalarthincod2}) as follows
\begin{align}
\boxs \, \pi &= 0 \label{fantasy1} \\[6mm]
2 \, M_6^4 \, \de_{\bar{\textbf{n}}}\evbhsi \, \pi + M_5^3 \, \boxfi \, \hp &= 0 \label{fantasy2} \\[2mm]
3 \, M_5^3 \, \hp^{\p}\Big\rvert_{0^+} &= \Bigg[ 4 \, M_6^4 \, \hdvfn^{\p} - 4 \, M_6^4 \, \sin \bigg( \frac{\bla}{2 \Msf} \bigg) \frac{\big( \de_z^2 - \de_y^2 \big)}{\boxf} \, \pi\evbbhsi  \Bigg]_{0^+} \label{fantasy3}
\end{align}
and
\begin{align}
\boxfi \, \hdvfn &= \half \, \de_{\bar{\textbf{n}}}\evbhsi \, \pi + 2 \, \de^2_{\hxi} \, \bigg( \frac{\de_{\bar{\textbf{n}}}}{\boxf}\evbhsi \, \pi \bigg) \label{fantasy4} \\[4mm]
2 \, \Mft \, \tan \bigg( \frac{\bla}{4 \Msf} \bigg) \, \boxf \, \hdvfn\Big\rvert_{0^+} &= \boxf \, \mcal{T}^{(4)}_{de} \label{fantasy5}
\end{align}
We see that the first of the cod-2 junction conditions (\ref{scalarthincod2}) acts as boundary condition of the Neumann type on the side of the thin cod-2 brane for the field $\hp$, while the second of the cod-2 junction conditions (\ref{scalarthincod2}) acts as boundary condition of the Dirichlet type on the side of the thin cod-2 brane for the field $\hdvfn$. Moreover, the details of the internal structure of the cod-2 brane, which are encoded in the explicit form of the sequences $f_{^{[n]}}$, $\ep_{^{[n]}}$ and $\h{\mcalT}_{ab}^{_{[n]}}$, do not enter the thin limit equations and only the integrated quantities $\bla$ and $\mcal{T}^{_{(4)}}_{\m\n}$ are present. Therefore, at least for the scalar sector, the ansatz we proposed in section \ref{The geometric ansatz} gives rise to a well-defined system of differential equations, and the internal structure of the cod-2 brane does not play a role in the thin limit. Since the ansatz corresponds to a class of configurations where the gravitational field is finite on the cod-2 brane, we conclude that (for what concerns the scalar sector) gravity is indeed regularized in the nested branes with induced gravity set-up, and that the thin limit of the cod-2 ribbon brane inside the (already thin) cod-1 brane is well-defined (at least for first order perturbations around pure tension solutions).

It is interesting to comment on the interplay between the master variables $\pi$ and $\hdvfn$, and on how the presence of matter on the cod-2 brane sources the total field configuration according to the coupled system of equations for $\pi$ and $\hdvfn$. First of all, note that if there is no matter on the cod-2 brane then the configuration $\pi = 0$, $\hdvfn = 0$ is a solution of the system of equations (it is the background solution in fact). If we turn on the energy-momentum tensor on the cod-2 brane, $\h{\mcalT}$ forces the cod-2 brane to move (equation (\ref{fantasy5})): remember in fact from section \ref{The induced gravity part} that $\hdvfn\big\rvert_{0^+}$ is proportional to the bending of the cod-2 brane in the bulk $\d \!\b^z_{\infty}$. This movement acts as a boundary condition for the movement of the cod-1 brane (equation (\ref{fantasy4})), producing a non-trivial cod-1 bending profile. This profile necessarily has non-vanishing first $\hxi$-derivative on the side of the cod-2 brane, and this acts as a source for the metric master variable $\pi$, since it produces a non-trivial boundary condition for $\hp$ on the side of the cod-2 brane (equation (\ref{fantasy3})). As a consequence, a non-trivial profile for $\pi$ in the bulk and on the cod-1 brane is created (equations (\ref{fantasy1}) and (\ref{fantasy2})). The profile of $\hp$ on the cod-1 brane in turn acts as a source for $\hdvfn$ on the cod-1 brane (equation (\ref{fantasy4})), and so on.

\subsubsection{Tensor and vector sectors}

Regarding the tensor sector, the thin limit equations of motion read
\begin{align}
\boxs \, \mscr{H}_{\m\n} &= 0 & \phantom{=}&(\textrm{bulk}) \label{tensorbulk} \\[5mm]
2 \Msf \, \de_{\bar{\mathbf{n}}}\evbhsi \, \mscr{H}_{\m\n} + \Mft \, \boxfi \, \h{\mscr{H}}_{\m\n} &= 0 & \phantom{=}&(\textrm{pure cod-1 brane}) \label{tensorpurecod1brane} \\[2mm]
\Mft \, \de_{\hxi}\Big\rvert_{0^+} \h{\mscr{H}}_{\m\n} &= - \mscrT^{(4)}_{\m\n} & \phantom{=}&(\textrm{cod-2 brane}) \label{tensorcod2brane}
\end{align}
and we see that the situation is analogous to the scalar sector. The equations (\ref{tensorbulk}) and (\ref{tensorpurecod1brane}) are exactly the equations for the tensor sector in a 6D cod-1 DGP model, and need an additional condition at $\hxi = 0^{\pm}$ to give a unique solution if we assume that $\h{\mscr{H}}_{\m\n}$ is continuous but not derivable in $\hxi = 0$. This condition is provided from the cod-2 junction condition (\ref{tensorcod2brane}), which is a boundary condition of the Neumann type. Analogously to the scalar sector, the energy-momentum tensor excites the metric field by providing a non-trivial boundary condition at the side of the cod-2 brane, which then creates a non-trivial profile on the cod-1 brane and in the bulk via the equations (\ref{tensorbulk}) and (\ref{tensorpurecod1brane}). Also in this case the internal structure of the cod-2 brane does not play a role in thin limit, so also in the tensor sector gravity is regularized and the thin limit of the cod-2 ribbon brane inside the (already thin) cod-1 brane is well-defined.
%
%

The situation in the vector sector is instead peculiar. The thin limit equations of motion read in terms of the master variable

\beq
\label{vectorbulkeqdiscussion}
\left.
\begin{aligned}
\de_z \, \boxs \, \Phi_{\m} &= 0 \\[4mm]
\de_y \, \boxs \, \Phi_{\m} &= 0
\end{aligned}
\quad \right\} \qquad (\textrm{bulk})
\eeq
\vspace{3mm}
\beq
\label{vectorpurecod1jcdiscussion}
\left.
\begin{aligned}
2 M_6^4 \, \bigtriangleup_{2}\evbhsi \, \Phi_{\m} + M_5^3 \, \boxf \, \de_{\bar{\mathbf{n}}}\evbhsi \, \Phi_{\m} &= 0 \\[3mm]
\de_{\hxi} \, \bigg( 2 M_6^4 \, \h{\Phi}_{\m} - M_5^3 \, \de_{\bar{\mathbf{n}}}\evbhsi \, \Phi_{\m} \bigg) &= 0
\end{aligned}
\quad \right\} \qquad (\textrm{pure cod-1 brane})
\eeq
\vspace{6mm}
\beq
\label{vectorcod2jcdiscussion}
\de_{z}\Big\rvert_{\bvfd(0,\chd)} \, \Phi_{\m}  = \de_{y}\Big\rvert_{\bvfd(0,\chd)} \, \Phi_{\m} = 0 \qquad (\textrm{cod-2 brane})
\eeq
and the gauge-invariant vector modes of metric perturbations are linked to the master variable by
\beq
\label{vectormastervariablediscussion}
\mcalA_{i\m} = \ep_{i}^{\,\, j} \, \de_j \, \Phi_{\m}
\eeq

Note first of all that the cod-2 junction conditions (\ref{vectorcod2jcdiscussion}) imply that the normal derivative of $\Phi_{\m}$ vanish on both sides of the cod-2 brane
\beq
\label{heypresto}
\de_{\bar{\mathbf{n}}}\Big\rvert_{\bvfd(0^{+},\chd)} \, \Phi_{\m} = \de_{\bar{\mathbf{n}}}\Big\rvert_{\bvfd(0^{-},\chd)} \, \Phi_{\m} = 0
\eeq
On the other hand, the second of the pure cod-1 junction conditions (\ref{vectorpurecod1jcdiscussion}) implies that the quantity
\beq
2 M_6^4 \, \h{\Phi}_{\m}\big( \hxi, \chd \big) - M_5^3 \, \de_{\bar{\mathbf{n}}}\Big\rvert_{\bvfd(\hxi, \chd)} \, \Phi_{\m}
\eeq
is constant with respect to the coordinate $\hxi$: the relations (\ref{heypresto}) then imply that
\beq
\label{apnea1}
2 M_6^4 \, \h{\Phi}_{\m}\big( \hxi, \chd \big) - M_5^3 \, \de_{\bar{\mathbf{n}}}\Big\rvert_{\bvfd(\hxi, \chd)} \, \Phi_{\m} = 2 M_6^4 \, \Phi_{\m}^{(4)}\big( \chd \big)
\eeq
where $\Phi_{\m}^{(4)}$ is the master variable evaluated on the cod-2 brane (i.e.~at $\hxi = 0$). Moreover, the bulk equations (\ref{vectorbulkeqdiscussion}) imply that
\beq
\boxs \, \Phi_{\m} \big( \Xd \big) = F_{\m}(\xd)
\eeq
where $F_{\m}$ is some vector function of the 4D coordinates $\xd$ only. Using the relation $\bigtriangleup_{2} \, \Phi_{\m} = - \boxf \, \Phi_{\m} + F_{\m}$ in the first of the pure cod-1 junction conditions (\ref{vectorpurecod1jcdiscussion}), and applying the operator $1/\boxf$, we get
\beq
\label{apnea2}
2 M_6^4 \, \h{\Phi}_{\m}\big( \hxi, \chd \big) - M_5^3 \, \de_{\bar{\mathbf{n}}}\Big\rvert_{\bvfd(\hxi, \chd)} \, \Phi_{\m} = 2 M_6^4 \, \frac{1}{\boxf} \, F_{\m}(\chd)
\eeq
and consistency of the equation (\ref{apnea2}) with the equation (\ref{apnea1}) then implies
\beq
F_{\m}(\chd) = \boxf \, \Phi_{\m}^{(4)}(\chd)
\eeq
Therefore, the equations of motion (\ref{vectorbulkeqdiscussion})--(\ref{vectorcod2jcdiscussion}) for the master variable are equivalent to the system 
\begin{align}
\boxs \, \Phi_{\m}\big( \Xd \big) &= \boxf \, \Phi_{\m}^{(4)}(\xd) \label{ultima1} \\[3mm]
2 M_6^4 \, \h{\Phi}_{\m}\big( \hxi, \chd \big) - M_5^3 \, \de_{\bar{\mathbf{n}}}\Big\rvert_{\bvfd(\hxi, \chd)} \, \Phi_{\m} &= 2 M_6^4 \, \Phi_{\m}^{(4)}(\chd) \\[2mm]
\de_{z}\Big\rvert_{\bvfd(0,\chd)} \, \Phi_{\m}  = \de_{y}\Big\rvert_{\bvfd(0,\chd)} \, \Phi_{\m} &= 0 \label{ultima3}
\end{align}

This result implies first of all that the system of differential equations for the vector sector is compatible, since the configuration
\beq
\label{basta}
\Phi_{\m}\big( \Xd \big) = \mathscr{F}_{\m}(\xd)
\eeq
is a solutions of the equations (\ref{vectorbulkeqdiscussion})--(\ref{vectorcod2jcdiscussion}) for any choice of the vector function $\mathscr{F}_{\m}$ of the 4D coordinates. Secondly, the solution for the master variable is not unique: this is not surprising since we mentioned above that there is freedom to redefine $\Phi_{\m}$ by adding to it any function of the 4D coordinates without changing the values of the gauge invariant variables $\mcalA_{z \m}$ and $\mcalA_{y \m}$. In fact, all the configurations of the form (\ref{basta}) correspond to the same solution for $\mcalA_{z \m}$ and $\mcalA_{y \m}$, namely the trivial solution of the equations of motion where $\mcalA_{z \m}$ and $\mcalA_{y \m}$ vanish everywhere. To understand if the solution for $\mcalA_{z \m}$ and $\mcalA_{y \m}$ is unique, for a generic configuration $\Phi_{\m}$ we define $\mathscr{F}_{\m} \equiv \Phi_{\m}^{_{(4)}}$ and $\Psi_{\m}(\Xd) \equiv \Phi_{\m}(\Xd) - \Phi_{\m}^{_{(4)}} (\xd)$, so that $\Phi_{\m}$ takes the form
\beq
\label{form}
\Phi_{\m}(\Xd) = \mathscr{F}_{\m}(\xd) + \Psi_{\m}(\Xd)
\eeq
where $\Psi_{\m}$ vanish on the cod-2 brane, i.e.~$\Psi^{_{(4)}}_{\m}(\chd) = 0$. Plugging the expression (\ref{form}) in the system (\ref{ultima1})--(\ref{ultima3}), we obtain the following system of differential equations for the function $\Psi_{\m}(\Xd)$
\begin{align}
\boxs \, \Psi_{\m} &= 0 \label{assolutamenteultima1} \\[3mm]
M_5^3 \, \de_{\bar{\mathbf{n}}}\evbhsi \, \Psi_{\m} - 2 M_6^4 \, \h{\Psi}_{\m} &= 0 \\[2mm]
\de_{z}\Big\rvert_{\bvfd(0,\chd)} \, \Psi_{\m}  = \de_{y}\Big\rvert_{\bvfd(0,\chd)} \, \Psi_{\m} &= 0 \\[3mm]
\Psi_{\m}^{(4)} &= 0 \label{assolutamenteultima4}
\end{align}
whose unique solution compatible with the boundary conditions is $\Psi_{\m} = 0$. Therefore the solution for the master variable in the vector sector is \emph{always} of the form $\Phi_{\m} = \mathscr{F}_{\m}(\xd)$, where $\mathscr{F}_{\m}$ is some vector function of the 4D coordinates only. We conclude that in the thin limit the gauge invariant vector modes are zero not only on the cod-2 brane, but also on the cod-1 brane and in the bulk. This is again consistent with the fact that the gauge invariant vector modes are not sourced in the thin limit (cfr.~equation (\ref{nottoobad2})), but is actually a much stronger conclusion.  Note that this conclusion is completely independent of the details of the internal structure of the cod-2 brane. It follows that our ansatz on the singular behaviour of the perturbation fields is consistent and that the thin limit of the cod-2 brane is well-defined also in the vector sector.

\section{Conclusions}
\label{Conclusions}

In this paper, we studied the behaviour of weak gravitational fields in models where a 4D brane is embedded inside a 5D brane equipped with induced gravity, which in turn is embedded in a 6D spacetime. More precisely, we considered a specific regularization of the branes internal structures where the 5D brane can be considered thin with respect to the 4D one, and studied perturbations at first order around the solutions corresponding to pure tension source configurations on the thick 4D brane. To perform the perturbative analysis, we adopted a bulk-based approach and expressed the equations in terms of gauge invariant and master variables using a 4D scalar-vector-tensor decomposition. We then studied the behaviour of the equations of motion in the thin limit of the ``ribbon'' 4D brane inside the (already thin) 5D brane. We proposed an ansatz on the behaviour of the perturbation fields in this limit, which corresponds to configurations where gravity remains finite everywhere even when the 4D brane is thin, and showed that the equations of motion give rise to a consistent set of differential equations in the thin limit, from which the details of the internal structure of the 4D brane disappear.

We conclude that, at least when considering first order perturbations around pure tension configurations, the thin limit of the ``ribbon'' 4D brane inside the (already thin) 5D brane is well defined, and that gravity on the 4D brane is finite even in the thin limit. We also confirm that the induced gravity term on the cod-1 brane is crucial for gravity regularization: if we set $\Mft = 0$ then the cod-2 junction conditions (\ref{fantasy3}), (\ref{fantasy5}) and (\ref{tensorcod2brane}) are consistent only with a source configuration where the matter energy-momentum tensor vanishes, i.e.~with a pure tension perturbation. Therefore, if $\Mft = 0$, our ansatz is in general not consistent; to accomodate matter on the cod-2 brane, either the bending of the cod-1 brane has to diverge at the cod-2 brane either the bulk metric has to become singular there, or both. We can understand geometrically the role of the cod-1 induced gravity term in the phenomenon of gravity regularization by noting that it allows to support a generic energy-momentum tensor on the thin cod-2 brane still having a continuous embedding of the cod-1 brane in a smooth bulk metric.

It is straightforward to apply the formalism developed in this paper to the 6D Cascading DGP model, since it is enough to include a 4D induced gravity term $\Mfs \, G^{(4)}_{\m\n}$ (which assures that gravity is 4D at small scales) in the codimension-2 energy momentum tensor $T^{(4)}_{\m\n}$. We are currently studying the properties of the 6D Cascading DGP model following this route \cite{SbisaKoyama:inprep}.

\section*{Acknowledgments}
FS wishes to thank Paolo Creminelli for hospitality at the Abdus Salam's ICTP, Trieste, Italy where part of this work has been done. FS also acknowledges useful conversations with Gianmassimo Tasinato. FS and KK were supported by the European Research Council's starting grant. KK is supported by the UK Science and Technology Facilities Council grants number ST/K00090/1 and ST/L005573/1.

\newpage

\begin{appendix}

\section{Thin limit of the background}
\label{Thin limit of the background}

In this appendix we want to study in detail the thin limit of the background configurations. As we mention in the main text, the thin limit of the background is performed by considering a sequence of source configurations of the form
\beq
\label{angelicabackground}
\bar{T}^{[n]}_{ab}(\hxi, \chd) = - \d_{a}^{\, \, \m} \, \d_{b}^{\, \, \n} \, f_{[n]} \big( \hxi \big) \, \bla \,\, \e_{\m\n}
\eeq
where $f_{^{[n]}} \big( \hxi \big)$ is a sequence of even functions which satisfy
\begin{align}
\int_{-\infty}^{+\infty} f_{[n]} \big( \hxi \big) \, d \hxi &= 1 & f_{[n]} \big( \hxi \big) &= 0 \,\,\,\, \text{for} \,\,\,\, \abs{\hxi} \geq l^{[n]}_2
\end{align}
and $l^{_{[n]}}_2$ is a sequence of positive numbers that converges to zero: $l^{_{[n]}}_2 \rightarrow 0^+$ for $n \rightarrow + \infty$. The analysis of section \ref{Pure tension solutions} implies that there exist exact solutions for this class of sources such that the bulk, induced and double induced metrics are respectively the 6D, 5D and 4D Minkowski metric, while the embedding of the cod-1 brane is $n$-dependent and non-trivial
\beq
\label{Pattysequence}
\bar{\varphi}_{[n]}^{A}(\hxi, \chd) = \big( Z_{[n]}(\hxi), Y_{[n]}(\hxi), \ch^{\a} \big) 
\eeq
and the cod-2 embedding is $n$-independent and trivial
\beq
\label{Robertsequence}
\bar{\a}_{[n]}^{a}(\chi^{\cdot}) = \big( 0, \chi^{\a} \big) 
\eeq
In terms of the regulating function $\ep_{^{[n]}}(\hxi)$
\beq
\ep_{[n]}(\hxi) \equiv \int_{0}^{\hxi} f_{[n]} (\z) \, d \z 
\eeq
the solution for the slope function reads
\beq
S_{[n]}(\hxi) = \frac{\bla}{2 \Msf} \, \ep_{[n]} \big( \hxi \big)
\eeq
and the solution for the $\hxi$-derivative of the embedding functions reads
\begin{align}
\Zp_{[n]}(\hxi) &= \sin \bigg( \frac{\bla}{2 \Msf} \, \ep_{[n]}(\hxi) \bigg) \label{Balotelli1} \\[2mm]
\Yp_{[n]}(\hxi) &= \cos \bigg( \frac{\bla}{2 \Msf} \, \ep_{[n]}(\hxi) \bigg) \label{Balotelli2}
\end{align}
Considering the slope function, we note that for $\hxi \gtrless \pm l_{2}^{_{[n]}}$ we have $\ep_{^{[n]}}(\hxi) = \pm \frac{1}{2}$, and by symmetry we have $\ep_{^{[n]}}(0) = 0$. For every fixed value $\hxi$ different from zero (say positive, although the case $\hxi < 0$ is analogous), there exists a natural number $N$ such that, for $n \geq N$, we have $l^{_{[n]}}_2 < \hxi$ (as a consequence of $l^{_{[n]}}_2 \rightarrow 0$) and so $\ep_{^{[n]}}(\hxi) = 1/2$. Therefore we have 
\beq
S_{[n]}(\hxi > 0) \xrightarrow[n \rightarrow + \infty]{} \frac{\bla}{4 \Msf}
\eeq
and so $S_{^{[n]}}$ converges to the function $S_{\infty}$ which reads
\beq
\label{negoziosottocasa}
S_{\infty} \big( \hxi \big) =
\begin{cases}
\frac{\bla}{4 \Msf} & \text{for $\hxi > 0$} \\
0 & \text{for $\hxi = 0$} \\
- \frac{\bla}{4 \Msf} & \text{for $\hxi < 0$}
\end{cases}
\eeq
and in particular $S_{[n]}\big( \pm l_{2}^{[n]} \big) = \pm \frac{\bla}{4 \Msf} \equiv S_{\pm}$ independently of $n$.

\subsection{Thin limit of the embedding functions}

To obtain the thin limit of the embedding functions, we integrate the relations (\ref{Balotelli1}) and (\ref{Balotelli2}) and we impose the condition $Z_{[n]}(0) = Y_{[n]}(0) = 0$ to get
\begin{align}
Z_{[n]}(\hxi) &= \int_{0}^{\hxi} \sin \bigg( \frac{\bla}{2 \Msf} \, \ep_{[n]}(\z) \bigg) d \z \label{Balotellibis1} \\[2mm]
Y_{[n]}(\hxi) &= \int_{0}^{\hxi} \cos \bigg( \frac{\bla}{2 \Msf} \, \ep_{[n]}(\z) \bigg) d \z \label{Balotellibis2}
\end{align}
For $\hxi \gtrless \pm l_2^{_{[n]}}$, we use the identity $\sin (\ep_{^{[n]}} \bla/2 \Msf) = \pm \sin (\bla/4 \Msf) + \big( \sin (\ep_{^{[n]}} \bla/2 \Msf) \mp \sin (\bla/4 \Msf) \big)$ and the linearity of the integral to get
\begin{align}
Z_{[n]}(\hxi) &= \sin \bigg( \frac{\bla}{4 \Msf} \bigg) \, \abs{\hxi} + Z_{[n]}^0 & Z_{[n]}^0 &= \int_{0}^{l_2^{[n]}} \sin \bigg( \frac{\bla}{2 \Msf} \, \ep_{[n]}(\z) \bigg) d \z - l_2^{[n]} \, \sin \bigg( \frac{\bla}{4 \Msf} \bigg)
\end{align}
Analogously, for $\hxi \gtrless \pm l_2^{_{[n]}}$ we get
\begin{align}
Y_{[n]}(\hxi) &= \cos \bigg( \frac{\bla}{4 \Msf} \bigg) \, \hxi \pm Y_{[n]}^0 & Y_{[n]}^0 &= \int_{0}^{l_2^{[n]}} \cos \bigg( \frac{\bla}{2 \Msf} \, \ep_{[n]}(\z) \bigg) d \z - l_2^{[n]} \, \cos \bigg( \frac{\bla}{4 \Msf} \bigg)
\end{align}
and taking the $l_2^{_{[n]}} \to 0$ limit we obtain
\begin{align}
Z_{\infty}(\hxi) &= \sin \bigg( \frac{\bla}{4 \Msf} \bigg) \, \abs{\hxi} & Y_{\infty}(\hxi) &= \cos \bigg( \frac{\bla}{4 \Msf} \bigg) \, \hxi
\end{align}
which is valid for $\hxi \in \mathbb{R}$.

Regarding the $\hxi$-derivative of the embedding functions, consistently with the symmetry properties of $\Zp$ and $\Yp$ we have
\begin{align}
\Zp_{[n]}(0) &= 0 & \Yp_{[n]}(0) &= 1
\end{align}
independently of $n$. In complete analogy to the case of the slope function, for every fixed value $\hxi > 0$ we have
\begin{align}
\Zp_{[n]}(\hxi > 0) &\xrightarrow[n \rightarrow + \infty]{} \sin \bigg( \frac{\bla}{4 \Msf} \bigg) & \Yp_{[n]}(\hxi > 0) &\xrightarrow[n \rightarrow + \infty]{} \cos \bigg( \frac{\bla}{4 \Msf} \bigg) 
\end{align}
and so we conclude that $\Zp_{^{[n]}}$ and $\Yp_{^{[n]}}$ converge respectively to the functions $\Zp_{\infty}$ and $\Yp_{\infty}$ which explicitly read
\beq
\label{felicita1}
\Zp_{\infty} \big( \hxi \big) =
\begin{cases}
\sin \Big( \bla/4 \Msf \Big) & \text{for $\hxi > 0$} \\
0 & \text{for $\hxi = 0$} \\
- \sin \Big( \bla/4 \Msf \Big) & \text{for $\hxi < 0$}
\end{cases}
\eeq
and
\beq
\label{felicita2}
\Yp_{\infty} \big( \hxi \big) =
\begin{cases}
\cos \Big( \bla/4 \Msf \Big) & \text{for $\hxi \neq 0$} \\
1 & \text{for $\hxi = 0$} 
\end{cases}
\eeq
Note that the sequences of functions $\Zp_{^{[n]}}$ and $\Yp_{^{[n]}}$ converge to \emph{discontinuous} functions. This is crucial for their convergence properties: in fact, since $\Zp_{^{[n]}}$ and $\Yp_{^{[n]}}$ are by hypothesis smooth for every value of $n$, if their convergence to $\Zp_{\infty}$ and $\Yp_{\infty}$ were uniform then $\Zp_{\infty}$ and $\Yp_{\infty}$ would necessarily be continuous (at least). The fact that $\Zp_{\infty}$ and $\Yp_{\infty}$ are discontinuous implies that the convergence of $\Zp_{^{[n]}}$ and $\Yp_{^{[n]}}$ is pointwise but not uniform.

\subsection{Thin limit of the parallel and normal vectors}

The sequences which correspond to the background parallel vector $\bar{\mathbf{v}}$ and to the background normal vector $\bar{\mathbf{n}}$ are expressed in terms of $\Zp_{^{[n]}}$ and $\Yp_{^{[n]}}$ as
\begin{align}
\bv^{A}_{[n]} (\hxi) &= \Big( \Zp_{[n]}, \Yp_{[n]}, 0,0,0,0 \Big) & \bn_{[n]}^A (\hxi) &= \Big( \Yp_{[n]}, -\Zp_{[n]}, 0,0,0,0 \Big)
\end{align}
The results obtained above imply that $\bar{\mathbf{v}}_{^{[n]}}$ and $\bar{\mathbf{n}}_{^{[n]}}$ respectively converge to the vector fields $\bar{\mathbf{v}}_{\infty}$ and $\bar{\mathbf{n}}_{\infty}$ which are separately constant for $\hxi > 0$ and for $\hxi < 0$, and which are discontinuous across the cod-2 brane
\begin{align}
\Big[ \bv^{z}_{\infty} \Big]_{0^{\pm}} &= 2 \, \sin \bigg( \frac{\bla}{4 \Msf} \bigg)  & \Big[ \bv^{y}_{\infty} \Big]_{0^{\pm}} &= 0 & \Big[ \bn^{z}_{\infty} \Big]_{0^{\pm}} &= 0  & \Big[ \bn^{y}_{\infty} \Big]_{0^{\pm}} &= -2 \, \sin \bigg( \frac{\bla}{4 \Msf} \bigg)
\end{align}
and in particular we have
\begin{align}
\bv^{i}_{\infty}\Big\rvert_{\hxi > 0} &= \Big( \sin \left( \bla/4 \Msf \right), \cos \big( \bla/4 \Msf \big) \Big) & \bn^{i}_{\infty}\Big\rvert_{\hxi  > 0} &= \Big( \cos \big( \bla/4 \Msf \big), -\sin \big( \bla/4 \Msf \big) \Big)
\end{align}

Regarding the parallel and normal components of the bending and of the perturbation of the parallel vector, the definitions (\ref{deltavdefinitiongi}) imply that for $\abs{\hxi} > l_2^{_{[n]}}$ we have
\begin{align}
\dhvsp^{[n]} &= \hdvfp^{[n] \, \p} & \dhvn^{[n]} &= \hdvfn^{[n] \, \p}
\end{align}
since for $\abs{\hxi} > l_2^{_{[n]}}$ we have $\bv_{i}^{_{[n] \, \p}} = 0$ and $\bn_{i}^{_{[n] \, \p}} = 0$. This implies that in the thin limit for $\hxi \neq 0$ we have
\begin{align}
\dhvsp^{\infty} &= \hdvfp^{\infty \, \p} & \dhvn^{\infty} &= \hdvfn^{\infty \, \p}
\end{align}
and in particular
\begin{align}
\label{susini}
\dhvsp^{\infty}\Big\rvert_{0^+} &= \hdvfp^{\infty \, \p}\Big\rvert_{0^+} & \dhvn^{\infty}\Big\rvert_{0^+} &= \hdvfn^{\infty \, \p}\Big\rvert_{0^+}
\end{align}

\section{Codimension-1 induced metric and curvature tensors}
\label{IndMetrandCurvTensApp}

In this appendix we give the explicit form of the perturbation of the cod-1 induced metric, of the cod-1 Ricci tensor and of the cod-1 extrinsic curvature in terms of the gauge invariant variables. We also express the cod-1 extrinsic curvature in terms of the normal and parallel components of the bending modes, and we discuss the apparent contradiction related to the presence of the parallel component $\dvfp$ in the extrinsic curvature.

\subsection{Induced metric}

Using the general definition $\ti{\mathbf{g}} \equiv \varphi_{\star} \big( \mathbf{g} \big)$, the perturbation of the induced metric can be decomposed as follows
\beq
\label{sunset}
\tilde{h}_{ab}(\xi^{\cdot}) = \tilde{h}^{[mp]}_{ab}(\xi^{\cdot}) + \tilde{h}^{[bp]}_{ab}(\xi^{\cdot})
\eeq
where
\begin{align}
\tilde{h}^{[mp]}_{ab}(\xi^{\cdot}) & \equiv \frac{\de \bvf^{A}(\xi^{\cdot})}{\de \xi^{a}} \, 
\frac{\de \bvf^{B}(\xi^{\cdot})}{\de \xi^{b}} \,\, h_{AB} \Big\rvert_{X^{\cdot} = \bvf^{\cdot}(\xi^{\cdot})} \\[2mm]
\tilde{h}^{[bp]}_{ab}(\xi^{\cdot}) & \equiv 
\frac{\de \dvf^{A}(\xi^{\cdot})}{\de \xi^{a}} \, \frac{\de \bvf^{B}(\xi^{\cdot})}{\de \xi^{b}} \,\, \e_{AB} + \frac{\de \bvf^{A}(\xi^{\cdot})}{\de \xi^{a}} \, \frac{\de \dvf^{B}(\xi^{\cdot})}{\de \xi^{b}} \,\, \e_{AB}
\end{align}
These two components have a clear geometrical meaning: the ``metric perturbation'' component $\tilde{h}^{[mp]}_{ab}$ represents the effect on the induced metric of the perturbation of the bulk metric, while the ``bending perturbation'' component $\tilde{h}^{[bp]}_{ab}$ represents the effect on the induced metric of the perturbation of the brane embedding.

The three sectors (scalar, vector and tensor) contribute additively to the perturbation of the induced metric, and we can write
\beq
\ti{h}_{ab} = \ti{h}^{(s)}_{ab} + \ti{h}^{(v)}_{ab} + \ti{h}^{(t)}_{ab}
\eeq
where $\ti{h}^{(s)}_{ab}$, $\ti{h}^{(v)}_{ab}$ and $\ti{h}^{(t)}_{ab}$ respectively contain only the scalar, vector and tensor sector of perturbation fields. These three contributions to the perturbation of the metric induced on the cod-1 brane read explicitly
\begin{align}
\ti{h}^{(s)}_{\xi\xi} &= \bv^{i} \bv^{j} \, \ti{h}^{gi}_{ij} + 2 \, \bv_i \, \dvf_{gi}^{i \, \p} \label{color1} \\[2mm]
\ti{h}^{(s)}_{\xi\m} &= \de_{\xi^{\m}} \Big( \bv_i \, \dvf_{gi}^{i} + \dvf_4^{gi \, \p} \Big) \\[2mm]
\ti{h}^{(s)}_{\m\n} &= \tp \, \e_{\m\n} + \deximxin \, 2 \, \dvf^{gi}_4 \\[5mm]
\ti{h}^{(v)}_{\xi\m} &= \dvf_{\m}^{_{(T)} \, \p} + \bv^i \, \mcalA_{i\m} \label{Color} \\[2mm]
\ti{h}^{(v)}_{\m\n} &= \de_{\xi^{( \m}} \, \dvf_{\n)}^{_{(T)}} \\[5mm]
\ti{h}^{(t)}_{\m\n} &= \ti{\mscr{H}}_{\m\n} \label{color6}
\end{align}
where the components not shown above vanish identically.

\subsection{Ricci and Einstein tensors}
\label{Ricci and Einstein tensors}

The geometric decomposition (\ref{sunset}) of $\ti{h}_{ab}$ implies that the perturbation of the cod-1 Ricci tensor can be decomposed as
\beq
\d \! \ti{R}_{ab}(\xi^{\cdot}) \simeq \d \! \ti{R}^{[mp]}_{ab}(\xi^{\cdot}) + \d \! \tilde{R}^{[bp]}_{ab}(\xi^{\cdot})
\eeq
where $\d \! \ti{R}^{[mp]}_{ab}$ is built from the metric perturbation part of the induced metric $\ti{h}^{[mp]}_{ab}$ and $\d \! \ti{R}^{[bp]}_{ab}$ is built from the bending perturbation part of the induced metric $\ti{h}^{[bp]}_{ab}$.

In analogy to the induced metric, also the perturbation of the Ricci tensor on the cod-1 brane can be written in the form
\beq
\d \! \ti{R}_{ab} = \d \! \ti{R}^{(s)}_{ab} + \d \! \ti{R}^{(v)}_{ab} + \d \! \ti{R}^{(t)}_{ab}
\eeq
where $\d \! \ti{R}^{(s)}_{ab}$, $\d \! \ti{R}^{(v)}_{ab}$ and $\d \! \ti{R}^{(t)}_{ab}$ respectively contain only the scalar, vector and tensor sector of perturbation fields. These three contributions to the perturbation of the Ricci tensor constructed with the metric induced on the cod-1 brane read explicitly

\begin{align}
\d \! \ti{R}^{(s)}_{\xi \xi} &= - \frac{1}{2} \, \bv^{i} \bv^{j} \, \, \boxf \, \ti{h}^{gi}_{ij} - 2 \, \tp^{\p\p} + \bv_{i}^{\p} \, \boxf \, \dvf^{i}_{gi} \\[2mm]
\d \! \ti{R}^{(s)}_{\m \xi} &= - \, \frac{3}{2} \, \de_{\xi} \, \dexim \, \tp \\[2mm]
\d \! \ti{R}^{(s)}_{\m \n} &= \deximxin \, \Big[ \!- \frac{1}{2} \, \bv^{i} \bv^{j} \, \ti{h}^{gi}_{ij} - \tp + \bv_{i}^{\p} \, \dvf^{i}_{gi} \Big] + \eta_{\m\n} \, \Big[ - \frac{1}{2} \, \Box_{5} \, \tp \, \Big] \\[4mm]
\d \! \ti{R}^{(v)}_{\xi\m} &= - \half \, \boxf \, \bv^i \, \ti{\mcalA}_{i\m} \\[2mm]
\d \! \ti{R}^{(v)}_{\m\n} &= \half \, \de_{\xi^{( \m}} \, \de_{\xi} \Big( \bv^i \, \ti{\mcalA}_{i|\n)} \Big) \\[4mm]
\d \! \ti{R}^{(t)}_{\m\n} &= - \half \, \boxfi \, \ti{\mscr{H}}_{\m\n}
\end{align}
where again the components not shown above vanish identically. Regarding the perturbation of the cod-1 Einstein tensor, we have also in this case
\beq
\d \! \ti{G}_{ab} = \d \! \ti{G}^{(s)}_{ab} + \d \! \ti{G}^{(v)}_{ab} + \d \! \ti{G}^{(t)}_{ab}
\eeq
where the vector and tensor parts satisfy
\begin{align}
\d \! \ti{G}^{(v)}_{ab} &= \d \! \ti{R}^{(v)}_{ab} & \d \! \ti{G}^{(t)}_{ab} &= \d \! \ti{R}^{(t)}_{ab}
\end{align}
and the scalar part reads explicitly
\begin{align}
\d \! \ti{G}^{(s)}_{\xi \xi} &= \frac{3}{2} \, \boxf \, \tp \label{porticciolo 1} \\[2mm]
\d \! \ti{G}^{(s)}_{\m \xi} &= - \, \frac{3}{2} \, \de_{\xi} \, \dexim \, \tp \label{porticciolo 2} \\[2mm]
\d \! \ti{G}^{(s)}_{\m \n} &= \deximxin \, \Big[ \!- \frac{1}{2} \, \bv^{i} \bv^{j} \, \ti{h}^{gi}_{ij} - \tp + \bv_{i}^{\p} \, \dvf^{i}_{gi} \Big] + \eta_{\m\n} \, \Big[ \, \frac{1}{2} \, \bv^{i} \bv^{j} \, \boxf \, \ti{h}^{gi}_{ij} + \boxf \, \tp + \frac{3}{2} \, \tp^{\p\p}
- \bv_{i}^{\p} \, \boxf \, \dvf^{i}_{gi} \, \Big] \label{porticciolo 3} 
\end{align}

\subsection{Extrinsic curvature}
\label{Extrinsic curvature}

Expanding at first order the formulae (\ref{ExtrCurvDecomposition})--(\ref{ExtrCurvB}) around the pure tension solutions, we get the following geometrical decomposition for the perturbation of the extrinsic curvature
\beq
\d \! \ti{K}_{ab}(\xi^{\cdot}) \simeq \d \! \ti{K}^{[og]}_{ab}(\xi^{\cdot}) + \d \! \ti{K}^{[pg]}_{ab}(\xi^{\cdot}) + \d \! \ti{K}^{[np]}_{ab}(\xi^{\cdot}) + \d \! \ti{K}^{[bp]}_{ab}(\xi^{\cdot})
\eeq
where
\begin{align}
\d \! \ti{K}^{[og]}_{ab}(\xi^{\cdot}) &\equiv - \half \, \frac{\de \bar{\varphi}^{A}(\xid)}{\de \xi^{a}}
\frac{\de \bar{\varphi}^{B}(\xid)}{\de \xi^{b}} \, \bar{n}^L(\xid) \frac{\de \, h_{AB}}{\de X^L}\Big\rvert_{X^{\cdot} = \bar{\vf}^{\cdot}(\xid)} \\[4mm]
\d \! \ti{K}^{[pg]}_{ab}(\xi^{\cdot}) &\equiv \half \, \bar{n}^A(\xid) \, \frac{\de \bar{\varphi}^{B}(\xid)}{\de \xi^{(a}} \, \frac{\de \bar{\varphi}^{L}(\xid)}{\de \xi^{b)}} \frac{\de \, h_{AB}}{\de X^L}\Big\rvert_{X^{\cdot} = \bar{\vf}^{\cdot}(\xid)}
\end{align}
\begin{align}
\d \! \ti{K}^{[np]}_{ab}(\xi^{\cdot}) &\equiv \half \, \bar{n}^{A}(\xid) \, \bar{n}^{B}(\xid) \, h_{AB}\evb \,\, \bar{n}_i(\xid) \, \bar{\varphi}^{i \, \p\p}(\xid) \\[4mm]
\d \! \ti{K}^{[bp]}_{ab}(\xi^{\cdot}) &\equiv \bar{n}_{L}(\xid) \, \frac{\de^2 \d \varphi^{L}(\xid)}{\de \xi^{a} \de \xi^{b}}
\end{align}
These four components have a clear geometrical interpretation: the orthogonal gradient $[og]$ component and the parallel gradient $[pg]$ component are non-zero respectively when the perturbation of the bulk metric has non-vanishing derivative in the orthogonal and in the parallel directions to the brane. The normal projection $[np]$ component is instead non-zero when the contraction of the perturbation of the bulk metric with two normal vectors is non-vanishing. Finally, the bending perturbation component $[bp]$ expresses the effect on the extrinsic curvature of the fact that the embedding of the brane is perturbed, even in absence of bulk metric perturbations.

Also in this case the perturbation of the extrinsic curvature can be written in the form
\beq
\d \! \ti{K}_{ab} = \d \! \ti{K}^{(s)}_{ab} + \d \! \ti{K}^{(v)}_{ab} + \d \! \ti{K}^{(t)}_{ab}
\eeq
where $\d \! \ti{K}^{(s)}_{ab}$, $\d \! \ti{K}^{(v)}_{ab}$ and $\d \! \ti{K}^{(t)}_{ab}$ respectively contain only the scalar, vector and tensor sector of the perturbation fields. These three contributions to the perturbation of the extrinsic curvature read explicitly

\begin{align}
\d \! \ti{K}^{(s)}_{\xi\xi} &= - \frac{1}{2} \, \bv^{i} \bv^{j} \, \de_{\bar{\mathbf{n}}}\evb \, h^{gi}_{ij}  + \bn^i \bv^{j} \, \de_{\bar{\mathbf{v}}}\evb \, h^{gi}_{ij} + \half \, \bn^i \bn^j \, \ti{h}^{gi}_{ij} \, \big( \bn^k \bv_{k}^{\p} \big) + \bn_i \, \dvf^{i \, \p \p}_{gi} \\[2mm]
\d \! \ti{K}^{(s)}_{\xi\m} &=  \dexim \, \Big[ \, \frac{1}{2} \, \bn^i \bv^{j} \, \ti{h}^{gi}_{ij} +  \bn_i \, \dvf^{i \, \p}_{gi} \, \Big] \\[2mm]
\d \! \ti{K}^{(s)}_{\m\n} &= - \frac{1}{2} \, \de_{\bar{\mathbf{n}}}\evb \, \pi \,\, \e_{\m\n} + \deximxin \,\, \bn_i \, \dvf^{i}_{gi} \\[4mm]
\d \! \ti{K}^{(v)}_{\xi\m} &= - \half \, \bv^i \, \de_{\bar{\mathbf{n}}}\evb \, \mcalA_{i\m} + \half \, \bn^i \, \ti{\mcalA}_{i\m}^{\p} \\[2mm]
\d \! \ti{K}^{(v)}_{\m\n} &= \half \, \de_{\xi^{( \m}} \, \bn^i \, \ti{\mcalA}_{i|\n)} \\[4mm]
\d \! \ti{K}^{(t)}_{\m\n} &= - \half \, \de_{\bar{\mathbf{n}}}\evb \, \mscr{H}_{\m\n}
\end{align}
where the components not shown above vanish identically.

\subsection{Some remarks on the bending modes}

We want now to comment on the number of gauge invariant variables which describe the fluctuation of the brane embedding. Note first of all that, as we mentioned in the main text, $\dvf_4^{gi}$ and $\dvf^{_{(T)}}_{\m}$ do not appear in $\d \! \ti{R}_{ab}$ and $\d \! \ti{K}_{ab}$. This implies that the 4D part of the bending $\dvf_{\m}$ does not appear in the equations of motion, and that the bending modes play a role only in the scalar sector. Regarding the other components of the bending modes, using the relation
\beq
\label{daje}
\bv_{i}^{\p} = \bvf_{i}^{\p\p} = \Sp \,\, \bn_i
\eeq
we can express the cod-1 Ricci tensor in terms of the normal and parallel bending modes $\dvfn$ and $\dvfp$ instead of $\dvf^{z}_{^{gi}}$ and $\dvf^{y}_{^{gi}}$. We get
\begin{align}
\d \! \ti{R}^{[bp]}_{\xi \xi} &= \Sp \,\, \boxf \, \dvfn \\[1mm]
\d \! \ti{R}^{[bp]}_{\m \n} &= \deximxin \, \Sp \,\, \dvfn
\end{align}
while $\d \! \ti{R}^{[bp]}_{\m \xi}$ vanishes. Therefore, the cod-1 Ricci tensor depends only on the normal component of the bending $\dvfn$, which is consistent with the fact that, from the point of view of the intrinsic geometry, $\dvfp$ represents just a change of coordinates. On the other hand, using the relation
\beq
\label{Bobbi}
\bn_{i}^{\p} = - \Sp \,\, \bv_i
\eeq
we obtain for the extrinsic curvature
\begin{align}
\d \! \ti{K}^{[bp]}_{\xi\xi} &= \dvfn^{\p\p} + \Spp \, \dvfp + 2 \, \Sp \, \dvfp^{\p} - \Spq \, \dvfn \\[1mm]
\d \! \ti{K}^{[bp]}_{\m\xi} &= \dexim \, \Big( \dvfn^{\p} + \Sp \, \dvfp \Big)  \\[1mm]
\d \! \ti{K}^{[bp]}_{\m\n} &= \deximxin \, \dvfn
\end{align}
and we see that the parallel component $\dvfp$ does not disappear from the extrinsic curvature. This seems somehow in contrast with the assertion of \cite{GarrigaVilenkin,Ishibashi:2002nn} that $\dvfn$ is the only physically observable fluctuation of the brane. To understand this apparent contradiction, note first of all that, from the point of view of the bulk, the normal component $\dvfn$ is the only perturbation which changes the shape of the brane, while $\dvfp$ just relabels the points on the brane: therefore, if we describe the extrinsic geometry with the second fundamental form $\mathbf{K}$, the parallel component does not play a role. However, the extrinsic curvature $\ti{\mathbf{K}} = \vf_{\star} (\mathbf{K})$ measures how fast the normal form changes when we move along a direction in the \emph{brane} coordinate system: a change of coordinates on the brane makes the normal vector vary more or less rapidly, and so has the same effect (from the brane point of view) as if we kept the brane coordinates unchanged and changed the shape of the brane (from the bulk point of view). Therefore, if we describe the extrinsic geometry with $\ti{\mathbf{K}}$, the mode $\dvfp$ does indeed play a role. The only exception to this argument is when the (background) normal form is constant, such as when the brane is straight in a homogeneous bulk metric: in this case, a change of coordinates on the brane has no effect on the normal form. In fact, if we set $\Sp = S^{\p\p} = 0$ in our case, the parallel component of the bending $\dvfp$ indeed disappears from the extrinsic curvature.

\section{The pillbox integration across the cod-2 brane}
\label{Pillbox integration across the cod2 brane}

In this appendix, we perform explicitly the pillbox integrations which appear in the left hand side of equations (\ref{Aloha der}) and (\ref{Aloha trace2}).

\subsection{A useful proposition}

In the appendix \ref{Thin limit of the background} we showed that the slope functions belonging to the sequence $\{S_{^{[n]}}\}_{^{[n]}}$ are smooth and odd with respect to the transformation $\hxi \to -\hxi$, and converge pointwise to the function (\ref{negoziosottocasa}). Furthermore, their value at the side of the cod-2 brane $\hxi = \pm l_2^{_{[n]}}$ is independent of $n$. This implies that $S^{\p}_{^{[n]}} \big( \hxi \big)$ is proportional to an even realization of the Dirac delta, since calling $\s \equiv \bla/ 4 \Msf$ we have
\beq
\int_{- l_2^{[n]}}^{+ l_2^{[n]}} \!\! d \hxi \,\, S^{\p}_{[n]} = 2 \, \s
\eeq
while $S^{\p}_{^{[n]}}$ vanishes for $\abs{\hxi} > l_2^{_{[n]}}$. Considering now a continuous real function $\mcal{B}$, we want first of all to see that $S^{\p}_{^{[n]}}(\hxi) \mcal{B} \big( S_{^{[n]}}(\hxi) \big)$ is as well a realization of the Dirac delta. More precisely, consider the following

\begin{lemma}
\label{lemma yes}
Be $\mcalB$ and $\mcalF$ continuous real functions, and define
\beq
\label{integral mon amour}
\mcalI_{B} \equiv \int_{-\s}^{+\s} \mcalB(\z) \, d\z
\eeq
Then we have
\beq
\label{lemmahypothesis}
\lim_{n \to \infty} \int_{- l_2^{[n]}}^{+ l_2^{[n]}} \! S^{\p}_{[n]}(\hxi) \, \mcal{B} \big( S_{[n]}(\hxi) \big) \, \mcalF(\hxi) \, d \hxi \,\, = \,\, \mcalI_{B} \, \mcalF(0) 
\eeq
\end{lemma}
\begin{proof}
Since $\mcalF(\hxi)$ is continuous on the compact interval $I_{^{[n]}} = \big[ - l_2^{_{[n]}}, l_2^{_{[n]}} \big]$ it has a maximum value $M_{^{[n]}}$ and a minimum value $m_{^{[n]}}$ in $I_{^{[n]}}$. Since $m_{^{[n]}} \leq \mcalF(\hxi) \leq M_{^{[n]}}$ on $I_{^{[n]}}$ we have
\beq
\lim_{n \to \infty} \int_{- l_2^{[n]}}^{+ l_2^{[n]}} \! S^{\p}_{[n]}(\hxi) \, \mcal{B} \big( S_{[n]}(\hxi) \big) \, \mcalF(\hxi) \, d \hxi \,\, \geq \,\, \lim_{n \to \infty} m_{[n]} \int_{- l_2^{[n]}}^{+ l_2^{[n]}} \! S^{\p}_{[n]}(\hxi) \, \mcal{B} \big( S_{[n]}(\hxi) \big) \, d \hxi \,\, = \,\, \mcalI_{B} \, \lim_{n \to \infty} m_{[n]}
\eeq
and analogously
\beq
\lim_{n \to \infty} \int_{- l_2^{[n]}}^{+ l_2^{[n]}} \! S^{\p}_{[n]}(\hxi) \, \mcal{B} \big( S_{[n]}(\hxi) \big) \, \mcalF(\hxi) \, d \hxi \,\, \leq \,\, \mcalI_{B} \, \lim_{n \to \infty} M_{[n]}
\eeq
Since the continuity of $\mcalF(\hxi)$ implies $\lim_{n \to \infty} M_{^{[n]}} = \lim_{n \to \infty} m_{^{[n]}} = \mcalF(0)$, we get (\ref{lemmahypothesis}).
\end{proof}

Therefore we can say that $S^{\p}_{^{[n]}}(\hxi) \mcal{B} \big( S_{_{[n]}}(\hxi) \big) \to \mcalI_{B} \, \d(\hxi)$. We now want to see that a result similar to (\ref{lemmahypothesis}) holds also if $\mcalF(\hxi)$ is substituted to a sequences of smooth functions which converges uniformly to a continuous function. More precisely we have
\begin{proposition}
\label{proposition yes}
Be $\mcalF_{[n]}(\hxi)$ a sequence of smooth functions which converge uniformly to a continuous function $\mcalF_{\infty}(\hxi)$. Then we have
\beq
\label{prophypothesis}
\lim_{n \to \infty} \int_{- l_2^{[n]}}^{+ l_2^{[n]}} \! S^{\p}_{[n]}(\hxi) \, \mcal{B} \big( S_{[n]}(\hxi) \big) \, \mcalF_{[n]}(\hxi) \, d \hxi \,\, = \,\, \mcalI_{B} \, \mcalF_{\infty}(0) 
\eeq
\end{proposition}
\begin{proof}
Since $\mcalF_{^{[n]}} \to \mcalF_{\infty}$ uniformly, for every $\vep > 0$ there exists a natural number $\ti{N}_{\vep}$ such that $n \geq \ti{N}_{\vep}$ implies $\abs{\mcalF_{^{[n]}}(\hxi) - \mcalF_{\infty}(\hxi)} \leq \vep$ for every $\hxi \in I_{^{[n]}}$. Then for $n \geq \ti{N}_{\vep}$ we have
\beq
\label{propres1}
\int_{- l_2^{[n]}}^{+ l_2^{[n]}} \! S^{\p}_{[n]}(\hxi) \, \mcal{B} \big( S_{[n]}(\hxi) \big) \, \mcalF_{[n]}(\hxi) \, d \hxi \,\, \geq \,\,
\int_{- l_2^{[n]}}^{+ l_2^{[n]}} \! S^{\p}_{[n]}(\hxi) \, \mcal{B} \big( S_{[n]}(\hxi) \big) \, \big( \mcalF_{\infty}(\hxi) - \vep \big) \, d \hxi
\eeq
while for the lemma \ref{lemma yes} there exists a natural number $\bar{N}_{\vep}$ such that $n \geq \bar{N}_{\vep}$ implies
\beq
\label{propres2}
\mcalI_{B} \, \mcalF_{\infty}(0) - \vep \,\, \leq \,\, \int_{- l_2^{[n]}}^{+ l_2^{[n]}} \! S^{\p}_{[n]}(\hxi) \, \mcal{B} \big( S_{[n]}(\hxi) \big) \, \mcalF_{\infty}(\hxi) \, d \hxi \,\, \leq \,\, \mcalI_{B} \, \mcalF_{\infty}(0) + \vep
\eeq
Calling $N_{\vep} = \textrm{max} \{ \ti{N}_{\vep} , \bar{N}_{\vep} \}$ and putting together (\ref{propres1}) and (\ref{propres2}) we get that $n \geq N_{\vep}$ implies
\beq
\label{propres3}
\int_{- l_2^{[n]}}^{+ l_2^{[n]}} \! S^{\p}_{[n]}(\hxi) \, \mcal{B} \big( S_{[n]}(\hxi) \big) \, \mcalF_{[n]}(\hxi) \, d \hxi \,\, \geq \,\, \mcalI_{B} \, \mcalF_{\infty}(0) - \vep \, \big( \mcalI_{B} + 1 \big)
\eeq
and analogously we get
\beq
\label{propres4}
\int_{- l_2^{[n]}}^{+ l_2^{[n]}} \! S^{\p}_{[n]}(\hxi) \, \mcal{B} \big( S_{[n]}(\hxi) \big) \, \mcalF_{[n]}(\hxi) \, d \hxi \,\, \leq \,\, \mcalI_{B} \, \mcalF_{\infty}(0) + \vep \, \big( \mcalI_{B} + 1 \big)
\eeq
Since $\vep > 0$ is arbitrary, (\ref{propres3}) and (\ref{propres4}) imply (\ref{prophypothesis}).
\end{proof}

\noi Moreover we have the following

\begin{corollary}
\label{corollary yes}
If, in addition to the hypothesis of Proposition \ref{proposition yes}, $\mcalF_{\infty}(\hxi)$ is odd with respect to the parity transformation $\hxi \to - \hxi$, then
\beq
\label{corollhypothesis}
\lim_{n \to \infty} \int_{- l_2^{[n]}}^{+ l_2^{[n]}} \! S^{\p}_{[n]}(\hxi) \, \mcal{B} \big( S_{[n]}(\hxi) \big) \, \mcalF_{[n]}(\hxi) \, d \hxi \,\, = \,\, 0
\eeq
\end{corollary}
\begin{proof}
Since by hypothesis $\mcalF_{\infty}(\hxi)$ is continuous, if it is odd then $\mcalF_{\infty}(0) = 0$. Then the thesis follows trivially from (\ref{prophypothesis}).
\end{proof}

\noi \textbf{Comment}: the results (\ref{prophypothesis}) and (\ref{corollhypothesis}) hold as well if the functions $\mcalF_{^{[n]}}$ and $\mcalF_{\infty}$ depend also on other variables (let's call them collectively $\chd$) beside $\hxi$, provided that, for every value of $\chd$, the functions $\mcalF_{^{[n]}}$ are smooth in $\hxi$ at fixed $\chd$ and converge to the continuous function $\mcalF_{\infty}$ uniformly with respect to $\hxi$ at fixed $\chd$. This is automatically granted if the functions $\mcalF_{^{[n]}}$ are smooth both in $\hxi$ and in $\chd$, and if they converge uniformly both with respect to $\hxi$ and to $\chd$ to a function $\mcalF_{\infty}$ which is continuous both in $\hxi$ and in $\chd$.

\subsection{The induced gravity part}
\label{The induced gravity part}

We turn now to the evaluation of the pillbox integral
\beq
\label{pillboxintegralindgrav}
\mscr{I}_{\!\!_{\textup{G}}} = \lim_{n \rightarrow + \infty} \int_{- l_{2}^{[n]}}^{+ l_{2}^{[n]}} \!\! d \hxi \,\,\, \bv_{i}^{[n] \, \p} \, \boxf \, \hdvf^{i}_{[n]}
\eeq
which appears in the equation (\ref{Aloha der}) and is produced by the cod-1 induced gravity term. Note that both sides of the equation above are functions of the 4D coordinates $\chd$, but we omit to indicate explicitly this dependence in the following. By expressing $\bv_i^\p \, \hdvf^{i}_{gi} = \big(\bv_i \, \hdvf^{i}_{gi}\big)^{\p} - \bv_i \, \hdvf^{i\, \p}_{gi}$ and remembering the definitions (\ref{deltavdefinitiongi}) it is easy to see that
\beq
\label{Dagmara}
\bv_i^{[n] \, \p} \, \hdvf^{i}_{[n]} = \hdvfp^{[n] \, \p} - \dhvsp^{[n]}
\eeq
and using this relation we can rewrite the integral (\ref{pillboxintegralindgrav}) as
\beq
\mscr{I}_{\!\!_{\textup{G}}} = \lim_{n \rightarrow + \infty} \int_{- l_{2}^{[n]}}^{+ l_{2}^{[n]}} \!\! d \hxi \,\,\, \boxf \, \hdvfp^{[n] \, \p} - \lim_{n \rightarrow + \infty} \int_{- l_{2}^{[n]}}^{+ l_{2}^{[n]}} \!\! d \hxi \,\,\, \boxf \, \dhvsp^{[n]}
\eeq
Since $\dhvsp^{[n]}$ is constructed from the first derivative of the embedding functions only, by our ansatz it remains bounded even in the $n \rightarrow + \infty$ limit. Therefore we have
\beq
\lim_{n \rightarrow + \infty} \int_{- l_{2}^{[n]}}^{+ l_{2}^{[n]}} \!\! d \hxi \,\,\, \boxf \, \dhvsp^{[n]} = 0
\eeq
and we conclude that
\beq
\mscr{I}_{\!\!_{\textup{G}}} = 2 \lim_{n \rightarrow + \infty} \boxf \, \hdvfp^{[n]}\Big\rvert_{l_{2}^{[n]}} = 2 \,\,\boxf \, \hdvfp^{\infty}\Big\rvert_{0^{+}}
\eeq
The parallel component of the bending, despite being non-zero outside the cod-2 brane, does not appear in the pure cod-1 junction conditions. However, in the thin limit its value in $\hxi = 0^+$ (on the side of the cod-2 brane) is not independent of the value of the normal component of the bending in $\hxi = 0^+$. In fact, the $y$ component of the bending $\hdvf^{y}_{\infty}$ vanishes in $\hxi = 0^+$ since it is continuous and odd; the definition (\ref{whoknows}) implies then that both $\hdvfnin$ and $\hdvfp^{\infty}$ on the side of the cod-2 brane are determined in terms of the 4D field $\hdvf^{z}_{\infty}\Big\rvert_{\hxi = 0}$ by
\begin{align}
\hdvfnin\Big\rvert_{\hxi = 0^+} &= \Yp_{\infty}\Big\rvert_{0^+} \, \hdvf^{z}_{\infty}\Big\rvert_{\hxi = 0} & \hdvfp^{\infty}\Big\rvert_{\hxi = 0^+} &= \Zp_{\infty}\Big\rvert_{0^+} \, \hdvf^{z}_{\infty}\Big\rvert_{\hxi = 0}
\end{align}
Therefore the integral (\ref{pillboxintegralindgrav}) can be equivalently expressed as
\beq
\mscr{I}_{\!\!_{\textup{G}}} = 2 \, \sin \bigg( \frac{\bla}{4 \Msf} \bigg) \, \boxf \, \hdvf^{z}_{\infty}\Big\rvert_{\hxi = 0}
\eeq
and
\beq
\mscr{I}_{\!\!_{\textup{G}}} = 2 \, \tan \bigg( \frac{\bla}{4 \Msf} \bigg) \, \boxf \, \hdvfn^{\infty}\Big\rvert_{0^{+}}
\eeq
It is worthwhile to stress that the values of $\hdvfnin$ and $\hdvfp^{\infty}$ are rigidly linked only on the side of the cod-2 brane and only in the thin limit, while for $\hxi \neq 0$ the two quantities are independent.

The field $\hdvf^{z}_{\infty}\big\rvert_{\hxi = 0}$ has a precise geometrical meaning. In general, the embedding of the (mathematical) cod-2 brane $\mcalC_2$ in the bulk $\b^{\cdot}(\chd)$ is obtained by composing the cod-2 embedding into the cod-1 brane $\ta^{\cdot}(\chd)$ and the cod-1 embedding $\vfd(\hxid)$
\beq
\b^A(\chd) = \vf^A \big( \ta^{\cdot}(\chd) \big)
\eeq
If we fix the embedding $\ta^{\cdot}(\chd)$ according to (\ref{Joan}) and we use cod-1 GNC, we have
\begin{align}
\b^z(\chd) &= \vf^z \big( 0, \chd \big) & \b^y(\chd) &= 0
\end{align}
since in this case $\vf^y$ vanishes at $\hxi = 0$ as a consequence of the cod-1 $\mathbb{Z}_2$ symmetry. Using the perturbative decomposition $\vf^z = Z + \hdvf^z$ and taking into account that $Z(0) = 0$, the bending of $\mcalC_2$ in the bulk in the extra dimensions $z$--$y$ is described by the function $\d \!\b^z ( \chd ) = \hdvf^z ( 0, \chd )$, whose gauge invariant version is
\beq
\d \!\b^z_{gi} \big( \chd \big) \equiv \hdvf^z_{gi}\big( 0, \chd \big)
\eeq
Since in the thin limit the physical ribbon brane and the mathematical brane $\mcalC_2$ coincide, the field $\hdvf^{z}_{\infty}\big\rvert_{\hxi = 0}$ is the gauge invariant bending $\d \!\b^z_{\infty}$ of the thin cod-2 brane in the bulk (in the extra dimensions $z$ and $y$).

\subsection{The extrinsic curvature part}

We consider now the pillbox integration
\beq
\label{pillboxintegralextcurv}
\mscr{I}_{\!\!_{\textup{K}}} = \lim_{n \rightarrow + \infty} \int_{-l_2^{[n]}}^{+l_2^{[n]}} \! d \hxi \, \bigg( S_{[n]}^{\p} \, \bn_{[n]}^i \bn_{[n]}^j \, \h{h}^{[n]}_{ij} + 2 \, \bn^{[n]}_i \, \hdvf^{i \, \p \p}_{[n]} \bigg)
\eeq
which appears in the equation (\ref{Aloha trace2}) and is produced by the extrinsic curvature term. Also in this case both sides of the equation are functions of the 4D coordinates $\chd$, but we omit to indicate explicitly this dependence. To perform this integration, it is useful to recast the integrand in a more convenient form. First of all, by expressing $\bn_i \, \hdvf^{i \,\p \p}_{gi} = \big(\bn_i \, \hdvf^{i \,\p}_{gi}\big)^{\p} - \bn_i^\p \, \hdvf^{i \,\p}_{gi}$ and using (\ref{Bobbi}) it is easy to see that (remember the definitions (\ref{deltavdefinitiongi}))
\beq
\bn^{[n]}_i \, \hdvf^{i \, \p\p}_{[n]} = \dhvn^{[n] \, \p} + S_{[n]}^{\p} \, \dhvsp^{[n]}
\eeq
so we can express the integrand as
\beq
\label{Montparnasse}
S_{[n]}^{\p} \, \Big( \bn_{[n]}^i \bn_{[n]}^j \, \h{h}^{[n]}_{ij} + 2 \, \dhvsp^{[n]} \Big) + 2 \, \dhvn^{[n] \, \p}
\eeq
Secondly, since in cod-1 GNC we have $\h{h}^{[n]}_{\xi\xi} = 0$, the relation (\ref{color1}) implies that
\beq
\label{Monmartre}
\dhvsp^{[n]} = - \half \, \bv_{[n]}^{i} \bv_{[n]}^{j} \, \h{h}_{ij}^{[n]}
\eeq
and inserting the relation (\ref{Monmartre}) in (\ref{Montparnasse}) we can express the integral (\ref{pillboxintegralextcurv}) as follows
\beq
\label{artisticequation}
\mscr{I}_{\!\!_{\textup{K}}} = 4 \, \dhvn^{\infty}\Big\rvert_{0^+} + \lim_{n \rightarrow + \infty} \int_{- l_{2}^{[n]}}^{+ l_{2}^{[n]}} \! d \hxi \, S_{[n]}^{\p} \, \Big( \bn_{[n]}^i \bn_{[n]}^j - \bv_{[n]}^{i} \bv_{[n]}^{j} \Big) \, \h{h}^{[n]}_{ij}
\eeq

The integral on the right hand side of (\ref{artisticequation}) is exactly of the form considered in the Proposition \ref{proposition yes}, since $\bv^i$ and $\bn^i$ are proportional to $\sin S$ or $\cos S$ and by our ansatz $\h{h}^{_{[n]}}_{ij}$ converge uniformly to the continuous functions $\h{h}^{\infty}_{ij}$. Moreover, the Corollary \ref{corollary yes} implies that $\h{h}^{_{[n]}}_{zy}$ gives a vanishing contribution to the thin limit integral, since it is odd; therefore the relation (\ref{artisticequation}) can be simplified as
\beq
\label{artisticequation2}
\mscr{I}_{\!\!_{\textup{K}}} = 4 \, \dhvn^{\infty}\Big\rvert_{0^+} + \lim_{n \rightarrow + \infty} \int_{- l_{2}^{[n]}}^{+ l_{2}^{[n]}} \! d \hxi \, S_{[n]}^{\p} \, \Big( \cos^2 \!S_{[n]} - \sin^2 \!S_{[n]} \Big) \, \Big( \h{h}^{[n]}_{zz} - \h{h}^{[n]}_{yy} \Big)
\eeq
Since in this case the integral defined in (\ref{integral mon amour}) reads
\beq
\mcalI \equiv \int_{-\s}^{+\s} \Big( \cos^2 \!\z - \sin^2 \!\z \Big) \, d\z = \int_{-\s}^{+\s} \cos (2\z) \, d\z = \sin 2 \s
\eeq
we finally get
\beq
\label{artisticequation3}
\mscr{I}_{\!\!_{\textup{K}}} = 4 \, \dhvn^{\infty}\Big\rvert_{0^+} + \sin \bigg(\frac{\bla}{2 \Msf}\bigg) \, \Big( \h{h}^{\infty}_{zz}(0) - \h{h}^{\infty}_{yy}(0) \Big)
\eeq
Also in this case, it is useful to express the integral $\mscr{I}_{\!\!_{\textup{K}}}$ in terms of $\dvfn^{\infty}$, which appears in the pure cod-1 junction conditions: using the relation (\ref{susini}) we obtain
\beq
\label{artisticequation5}
\mscr{I}_{\!\!_{\textup{K}}} = 4 \, \hdvfn^{\infty \, \p}\Big\rvert_{0^+} + \sin \bigg( \frac{\bla}{2 \Msf} \bigg) \, \Big( \h{h}^{\infty}_{zz} (0) - \h{h}^{\infty}_{yy} (0) \Big)
\eeq

\end{appendix}

\clearemptydoublepage

%
%

\end{document}